\documentclass[12pt,final,
thmsa,
epsf,a4paper]
{article}
\newtheorem{theorem}{Theorem}

\newtheorem{definition}{Definition}
\newtheorem{example}{Example}
\newtheorem{lemma}{Lemma}

\newenvironment{proof}[1][Proof]{\emph{#1.} }{\  \hfill $\square $ \vspace{5 pt}}

\usepackage{stackrel}
\usepackage{natbib}
\usepackage{amssymb}
\usepackage[dvips]{graphics}
\usepackage{amsmath}

\usepackage{cancel}

\usepackage{enumerate}

\usepackage{amsfonts}

\usepackage{amssymb}

\usepackage{enumerate}

\usepackage{mathrsfs}

\usepackage[usenames]{color}

\usepackage{pgf,tikz}

\usetikzlibrary{decorations.pathreplacing,angles,quotes}
\usetikzlibrary{arrows.meta}
\usetikzlibrary{arrows, decorations.markings}

\usepackage{hyperref}
\hypersetup{
    colorlinks,
    citecolor=blue,
    filecolor=black,
    linkcolor=blue,
    urlcolor=blue
}

\usepackage{booktabs}

\usepackage{soul}
\usepackage{pgf,tikz}
\usetikzlibrary{arrows}

\usepackage[utf8]{inputenc}
\usepackage{amssymb, amsmath,geometry}
\usepackage{fancyhdr}

\usepackage{array}

\usepackage{soul}

\usepackage{emptypage}

\newcommand*\samethanks[1][\value{footnote}]{\footnotemark[#1]}

\begin{document}

\title{Mechanisms for a dynamic many-to-many school choice problem\thanks{We are grateful to María Haydée Fonseca-Mairena, Marina Nuñez, Juan Pereyra, Francisco Robles, Oriol Tejada, William Thomson, Juan Pablo Torres Martínez, and Bumin Yenmez and participants of the Recent Advances in School Choice Workshop for their valuable comments and suggestions on earlier versions of this paper. 
We acknowledge financial support from the UNSL (grants 032016 and 031323), the CONICET (grant PIP112-200801-00655), and the Agencia I+D+i (grant PICT2017-2355).\smallskip}}

\author{Adriana Amieva\thanks{Instituto de Matem\'{a}tica Aplicada San Luis, Universidad Nacional de San Luis and CONICET, San Luis, Argentina;   e-mails: \texttt{avamieva@email.unsl.edu.ar} (A. Amieva), \texttt{abonifacio@unsl.edu.ar} (A. Bonifacio),  \texttt{paneme@unsl.edu.ar} (P. Neme).\smallskip} \and  Agustín G. Bonifacio\samethanks[2] \and Pablo Neme\samethanks[2]}

\date{\today}
\maketitle

\begin{abstract}
We examine the problem of assigning teachers to public schools over time when teachers have tenured positions and can work simultaneously in multiple schools. To do this, we investigate a dynamic many-to-many school choice problem where public schools have priorities over teachers and teachers hold path-independent choice functions selecting subsets of schools. We introduce a new concept of dynamic stability that recognizes the tenured positions of teachers and we prove that a dynamically stable matching always exists. We propose the Tenure-Respecting Deferred Acceptance $(TRDA)$ mechanism, which produces a dynamically stable matching that is constrained-efficient within the class of dynamically stable matchings and minimizes unjustified claims. To improve efficiency beyond this class, we also propose the  Tenure-Respecting Efficiency-Adjusted Deferred Acceptance $(TREADA)$ mechanism, an adaptation of the Efficiency-Adjusted Deferred Acceptance mechanism to our dynamic context. We demonstrate that the outcome of the $TREADA$ mechanism Pareto-dominates any dynamically stable matching and achieves efficiency when all teachers consent. Additionally, we examine the issue of manipulability, showing that although the $TRDA$ and $TREADA$ mechanisms can be manipulated, they remain non-obviously dynamically manipulable under specific conditions on schools' priorities.
\bigskip



\end{abstract}

\section{Introduction}\label{intro}

School choice models have been extensively studied in the literature on two-sided matching. These models can be categorized into two variants. On the one hand, there is the traditional school choice model, where parents rank the schools they want to send their children to according to their preferences, and schools, following local laws, establish priority rankings among children \citep[see for instance][among many others]{abdulkadirouglu2003school,Abdulkadiroglu2004,Abdulkadiroglu2005,Abdulkadiroglu2009,Abdulkadiroglu2011,kesten2010school,Klijn2013,Kojima2011}. On the other hand, there is a variant of the school choice model in which teachers search for jobs in public schools, and schools have priority rankings over teachers, as in the previous case. Our paper focuses on this second variant. We consider a school choice model in which each public school has a priority ranking established by local laws, and each teacher has a choice function selection subsets of schools \cite[see for instance][among others]{pereyra2013dynamic,chen2019self,tang2021weak,combe2022design}.

In this paper, we examine the problem of assigning teachers to public schools. The main distinctive feature of our model is that teachers can be assigned to more than one school. Additionally, we recognize that this problem is dynamic, as the teacher population changes from one period to the next. Specifically, given an initial assignment of teachers to public schools, some teachers choose to exit the market while others decide to enter. Our model assumes that teachers whom schools already employ hold tenured positions. Consequently, schools are unable to lay off teachers in order to hire new ones, and they can only do so when a vacant position appears.

An important characteristic of school choice models is that the state and local authorities determine school priorities. For instance, in many educational systems worldwide, there is a classification board (or examining board) that ranks teachers for each school. This ranking is developed based on several factors, such as teachers' CVs, the distance from teachers' houses to the school, the school’s orientation, among others. In this paper, we assume that the schools' priorities are strict.

In many countries, such as Argentina, teachers are employed by several schools. Depending on the characteristics of job offers, teachers may not have a ranking responsive to a linear order over individual schools.\footnote{The property of responsiveness ensures that a teacher’s preference over sets of schools is consistent with her individual ranking over schools. A teacher prefers one set of schools over another if the first set contains more individually preferred schools, and adding a desirable school to a set makes it strictly more preferred.} For instance, a teacher’s choice function may not be representable by a ranking over individual schools. She might choose to work at both schools $A$ and $B$ rather than only at $A$ or only at $B$. However, she may still choose school $C$ when it is available, as it offers better working conditions than the combination of $A$ and $B$. Furthermore, due to feasibility constraints—such as geographical distance—she may not consider working simultaneously at $A$ and $C$ or at $B$ and $C$ as viable options. In such a case, her choice function would select ${C}$ when ${A, B, C}$ are available; ${A, B}$ when only $A$ and $B$ are available; lastly either ${A}$ or ${B}$ when offered separately; and so on.\footnote{Note that if a teacher's choice function does not choose schools $A$ and $C$ combined when available nor schools $B$ and $C$ combined when available, then her choice will not be \emph{responsive} to the linear ordering of individual schools in which $A$ is first, $B$ is second, and $C$ is last. For a thorough treatment of responsive preferences see \cite{RothSotomayor90}.} This scenario suggests that the most appropriate way to model teachers' choice behavior is through path-independent choice functions. This modeling approach is justified by the result of \cite{aizerman1981general}, who show that a choice function is path-independent if and only if it satisfies two key properties: \emph{substitutability} and \emph{consistency}. Substitutability requires that, whenever a teacher selects a given subset of schools from a larger set of options, she must continue selecting each of those schools even if some of the other schools are no longer available. In other words, the inclusion of additional options should not be necessary for a school to be chosen. Consistency, in turn, ensures that if a smaller set of schools contains all those that were selected from a larger set, then the teacher's choice from the smaller set must remain unchanged. For a detailed discussion, see also \cite{chamb2017choice}.

When studying many-to-many school choice problems, \emph{stability} is the most relevant solution concept. There are three conditions needed for such a concept to be met. The first condition is \emph{individual rationality}, i.e. it should assign each teacher only to schools on her choice list or leave her unassigned.\footnote{Usually, for two-sided matching models, this condition is required for both sides of the problem. However, due to the nature of the school choice model, the condition for the schools’ side is implicitly embedded in the priorities.} The second condition requires that there is no unmatched teacher--school pair $(i, s)$ such that teacher $i$ \emph{has a claim} over school $s$---that is, (i) $i$ would prefer to be assigned to $s$, and (ii) some teacher currently assigned to $s$ has lower priority than $i$.\footnote{In school choice models, this concept is also known as \emph{justified envy} \citep[see][for more details]{abdulkadirouglu2003school}.}  Lastly, for a matching to be stable it must be \emph{non-wasteful}, i.e. whenever a teacher would like to add a school to her assignment, that school already has all its positions filled.
 \cite{martinez2008invariance} show that the set of stable matchings remains the same when an agent’s preference relation changes, as long as these two preference relations induce the same choice function. Since it is well known that the set of stable matchings is non-empty when preferences are substitutable, then, in our context, when firms have path-independent choice functions, the set of stable matchings is non-empty \cite[see][for more details]{chamb2017choice}.

When considering dynamic problems, we must adapt the notion of stability to such a dynamic context. In this paper, we define an appropriate notion of stability that takes into account the tenured nature of our model by considering the matching from the previous period for the teachers who decided to stay in the market. 
One of the notions we need to adapt is that of individual rationality. Following the same reasoning as in static problems, if a teacher is assigned to a school not on her choice list, she may seek another job outside the market. The same happens if she worsens her labor situation from one period to the next when she decides to stay in the market. In dynamic problems, we call a matching that prevents both situations \emph{dynamically rational}.

Teachers can have claims over schools at a matching in dynamic problems, as in static ones. Due to the tenured nature of teachers' jobs, we consider two kinds of claims. Remember that, when a teacher has a claim over a school, there is another teacher with lower priority currently assigned to that school. If the latter teacher was not assigned to the school in the previous period, the claim is \emph{justified}. Otherwise, the claim is \emph{unjustified}. Our adapted notion of stability for the dynamic environment says that a matching is \emph{dynamically stable} if it is dynamically rational, eliminates justified claims, and is non-wasteful.

For a dynamic problem, we present a mechanism---the Tenure-Respecting Deferred Acceptance mechanism ($TRDA$ from now on)---that respects the tenured nature of the problem and produces a dynamically stable matching as its outcome. In doing so, we prove that a dynamically stable matching always exists for each dynamic problem. Moreover, the outcome of the $TRDA$ mechanism improves upon any other dynamically stable matching from the teachers' standpoint.
 We also show that the $TRDA$ mechanism not only eliminates justified claims but also minimizes unjustified ones, i.e.,  
there is no other mechanism yielding a dynamically stable matching where the set of unjustified claims is a strict subset of the unjustified claims produced by the $TRDA$ mechanism.

When considering dynamic problems, a related paper is \cite{pereyra2013dynamic}. This paper examines a dynamic problem similar to ours, but with the restriction that teachers must work at most in one school. By using the teacher-proposing deferred acceptance mechanism for this less general market, \cite{pereyra2013dynamic} shows the existence of a dynamically stable matching that: (i) minimizes unjustified claims and (ii) Pareto-dominates any other dynamically stable matching. Our paper extends the results of \cite{pereyra2013dynamic} to a more general problem where teachers can work at multiple schools and also have path-independent choice functions.

Following the line of research on dynamic markets, recent work explores alternative formulations where agents arrive over time. \cite{doval2022dynamicallstable} introduces dynamic stability for two-sided, one-to-one markets, addressing timing issues and externalities from unmatched agents, ensuring stable matchings and timely participation. \cite{nicolo2023dynamic} propose the Dynamic Core for dynamic one-sided markets and design the Intertemporal Top Trading Cycle algorithm, which finds Pareto-efficient and group strategy-proof allocations within this core. Building on these ideas, \cite{saulle2024rationalizable} refine dynamic stability by incorporating agents’ rationalizable conjectures, defining a new stability concept that refines Doval’s and is always non-empty.

Beyond stability, efficiency and non-manipulation are desirable properties in mechanism design. On the one hand, it is known that there is a trade-off between stability and efficiency \citep{Roth1985} and, therefore, in static problems, the idea of improving efficiency while relaxing the stability condition has been widely explored. On the other hand,  requiring non-manipulability of mechanisms can be difficult to achieve, especially in general matching problems. However, it is common in the literature to relax this condition and require that manipulations of agents are not easily recognizable. We address both concerns in the context of dynamic markets. In addition to introducing the $TRDA$ mechanism, we analyze an alternative mechanism that achieves greater efficiency at the cost of dynamic stability. Furthermore, we study the strategic behavior of teachers under both mechanisms. While the non-manipulability of dynamic mechanisms cannot be guaranteed, we ensure that our two mechanisms are not easily manipulable.

The teacher-proposing deferred acceptance mechanism consistently selects the optimal stable matching for teachers in a static school choice problem. However, it is well-known that the teacher-proposing deferred acceptance mechanism is not always efficient \citep{Roth1985a}. This inefficiency problem has been empirically demonstrated, as shown by \cite{Abdulkadiroglu2009} in the case of New York City's high school match, where thousands of students could have been assigned to more preferred options without harming others.\footnote{In our many-to-many model, the role of the students is played by the teachers, but the same inefficiency prevails.}

For a standard many-to-one static school choice problem, \cite{kesten2010school} proposed the Efficiency-Adjusted Deferred Acceptance mechanism ($EADA$ from now on)  to enhance the efficiency of the teacher-proposing deferred acceptance mechanism beyond the class of stable matchings. The $EADA$ mechanism initially implements the teacher-proposing deferred acceptance mechanism and then iteratively identifies teachers who cause inefficiencies by applying to schools that will ultimately reject them. With these teachers' consent, their applications to such schools are subsequently blocked, ensuring that each teacher is assigned to either the same or a higher-ranked school compared to the original outcome of the teacher-proposing deferred acceptance mechanism.


The $EADA$ mechanism requests the consent of teachers who are likely to be rejected by schools to waive their priority at those schools. This consent is only requested when it does not harm the teachers' chances of being assigned to a better school, meaning that waiving their priority will not negatively affect their outcome. This ensures that teachers have no reason to withhold consent, as doing so does not disadvantage them and can improve the overall efficiency of the matching process by benefiting other teachers.
The implementation of the $EADA$ mechanism satisfies numerous desirable properties and emerges as a serious contender to the teacher-proposing deferred acceptance mechanism, particularly in contexts where both efficiency and stability are crucial.

Considering that our model is many-to-many, that teachers have path-independent choice functions and tenured positions, and given the dynamic nature of the problem, we propose the \textit{Tenure-Respecting Efficiency-Adjusted Deferred Acceptance} mechanism ($TREADA$ from now on), which generalizes the one introduced by \citet{kesten2010school}. We show that, when teachers have path-independent choice functions, the outcome of the $TREADA$ mechanism, although not always dynamically stable, improves upon any dynamically stable matching, including in particular the outcome of the $TRDA$ mechanism. Moreover, if all teachers consent, the outcome of the $TREADA$ mechanism is efficient.

We analyze the possibility of manipulation by teachers by adopting a framework in which they report substitutable preferences over subsets of schools, rather than path-independent choice functions. While this reformulation does not constitute a substantive change, expressing the problem in terms of preferences is more suitable for studying strategic behavior, as it allows manipulations to be modeled as deviations in the reported preferences. The equivalence between both formulations stems from the fact that path-independent choice functions are precisely those that are substitutable and consistent. Conversely, any substitutable preference relation induces a choice function that satisfies both substitutability and consistency \citep[see][for details]{martinez2008invariance}. Therefore, focusing on substitutable preferences preserves the essential structure of the problem.

Within this framework, it becomes possible to study whether teachers can benefit by misrepresenting their preferences—raising a central question in mechanism design: can the proposed mechanism be manipulated? And if so, how obvious is the manipulation? Following the result of \cite{martinez2004group} for static many-to-one matching problems, the teacher-proposing deferred acceptance mechanism may be subject to manipulation \citep[see also][among others]{hatfield2005matching,Abdulkadiroglu2009,kesten2010school,alva2019strategy}. In this setting, substitutability alone does not guarantee immunity from manipulation, and this insight extends naturally to many-to-many problems.

In our analysis, due to the structure of schools' priorities, only teachers have the capacity to behave strategically. This asymmetry in strategic possibilities has been explored extensively in many-to-many matching problems \citep[see][among others]{sakai2011note,hirata2017stable,iwase2022equivalence}.

Given that the complete elimination of manipulations is not feasible,  the literature has shifted to stable-dominating mechanisms that are robust against obvious manipulations.\footnote{A stable-dominating mechanism is either a stable mechanism or mechanisms that Pareto-dominates (from teachers' standpoint) a stable one.}  As defined in \cite{troyan2020obvious}, a manipulation is considered \emph{obvious}
if the best possible outcome under the manipulation is strictly better than the best possible
outcome under truth-telling, or if the worst possible outcome under the manipulation is strictly
better than the worst possible outcome under truth-telling. The immunity of these kinds of manipulations is studied in \cite{arribillaga2023obvious} for a static many-to-one problem with substitutable preferences, and in \cite{arribillaga2024many-yo-many-obvious} for a static many-to-many problem with substitutable preferences on both sides.\footnote{Another weakening of non-manipulability recently considered in the literature is that of \emph{regret-free truth-telling}, where no student regrets reporting her true preferences. \cite{chen2024regret} demonstrates that the $EADA$ mechanism is regret-free truth-telling for a many-to-one school choice market.}

Due to the temporal evolution of the population of teachers and their tenured conditions,  we adapt the notion of obvious manipulability to our dynamic context, coining it  \emph{obvious dynamic manipulability}. Assuming substitutability in teachers' preferences, without placing any restrictions on schools' priorities, we show that the $TRDA$ mechanism can be obviously dynamically manipulable. To guarantee that not only the $TRDA$ mechanism but also the $TREADA$ and all stable-dominating mechanisms are immune to obvious dynamic manipulations, a condition on schools' priorities is required. This condition, first introduced by \cite{pereyra2013dynamic}, is referred to as \emph{lexicographic by tenure} and says that, in each period, teachers who were matched in the previous period are given priority at every school over those entering the problem in the current period. 
The paper is organized as follows. Section \ref{sec preliminaries} presents the problem and preliminaries. The $TRDA$ mechanism and its properties are in Section \ref{section TRDA mecanismo}. Section \ref{seccion de eficiencia} presents the $TREADA$ mechanism that improves the efficiency over all dynamically stable mechanisms. The study of the non-obvious manipulability of the mechanisms presented in Sections \ref{section TRDA mecanismo} and \ref{seccion de eficiencia} is in Section \ref{Seccion de manipulation no obvia}. Section \ref{seccion de discusion} contains some final remarks. Finally,  Appendix \ref{apendice de pruebas} gathers all proofs.

\section{Preliminaries}\label{sec preliminaries}

We consider a dynamic many-to-many school choice problem where time is discrete and lasts forever (starting at $t=1$). At each period, there are two disjoint sets of agents: the set of schools (which will always be the same) and the set of teachers (which will vary from period to period).\footnote{From one period to the next, the set of teachers changes since teachers may leave or new teachers may enter the market.} We also assume that teachers hold \textit{tenured} positions throughout the periods, that is, once a school employs a teacher, the school cannot fire her. We denote by $S$ the set of schools and by $I^t$ the set of teachers in the problem at period $t$. Each school  $s\in S$ has, at period $t$, a strict priority relation $\boldsymbol{>_s^t}$ over the set of teachers available at this period, and is represented by the ordered list of teachers (from most to least preferred). We assume throughout the paper that the relative order between teachers in different periods is the same, i.e., given two periods $t$ and $t'$,  $s\in S$, and $i,j\in I^t\cap I^{t'}$, $i>_s^t j$ if and only if $i>_s^{t'} j.$  Each teacher $i\in I^t$ has a choice function $C_i:2^S\to 2^S$ that, given a subset of schools $S'$, $C_i(S')$ selects a subset of $S'$. Given a teacher $i\in I^t$, a school $s\in S$ is \textbf{acceptable for $\boldsymbol{i}$} if $s\in C_i(\{s\})$.  Teachers reveal their choice functions once they enter the problem.

We say that teacher $i$'s choice function satisfies \textbf{path-independence} if 
\begin{equation}\label{propiedad 1}
C_i(S'\cup S'')=C_i(C_i(S')\cup S'')
\end{equation} for each pair of subsets $S'$ and $S''$ of schools. \cite{aizerman1981general} establishes that a choice function is path-independent if and only if it satisfies \textbf{substitutability} and \textbf{consistency} \citep[see also][ for more details]{chamb2017choice}. Substitutability states that if for each $ S'\subseteq S$ we have that $s\in C_i(S')$ implies that $s\in C_i(S'\setminus \{s'\})$  for each $s'\in S'\setminus \{s\}$. Consistency states that for $S',S''\in ,$ if $C_i\subseteq S''\subseteq S'$ then $C_i(S')=C_i(S'')$. 
Throughout the paper, we assume that teachers have path-independent choice functions.

 Given teacher $i$'s Choice function $C_i$, \cite{Blair1988} defines a partial order for $i$  over subsets of schools as follows: given two subsets of schools  $S'$ and $S''$,  $\boldsymbol{S' \succeq^B_i S''}$ whenever $S'=C_i(S' \cup S'')$. We also write  $\boldsymbol{S'\succ^B_i S''}$ whenever $S' \succeq^B_i S''$ and $S' \neq S''$. 
This partial order over sets, under the assumption of substitutability, was first used by \citet{Roth1984b} to show the optimality of the \textit{DA} mechanism, and later by \citet{Blair1988} to prove that the set of stable matchings forms a lattice. Since then, many papers in the literature have employed this order (known as Blair's partial order) to deal with substitutability \cite[see][among others]{Fleiner2003,Echenique2004,Echenique2006}.

A \textbf{static problem} is a tuple $$(I, S , C_{I}, >_S,q_S)$$ where $I$ is the set of teachers, $S$ is the set of schools, $C_{I}= (C_i)_{i\in I}$ is the profile of choice functions of teachers and  $>_S^t=(>_s^t)_{s\in S}$ is the profile of priorities of schools, and $q_S=(q_s)_{s\in S}$ is the profile of quotas of schools.

A \textbf{matching} $\mu$ for a static problem is a function from the set $I\cup S$ in $2^{I \cup S}$ such that for each $i\in I$ and for each $s\in S$ fulfills the following conditions: (i) $\mu(i)\subseteq S$, (ii) $\mu(s)\subseteq I$, (iii) $|\mu(s)|\leq q_s$, and (iv) $i\in\mu(s) \text{ if and only if }s\in\mu(i).$

 Given a static problem, we say that teacher $i$ has a \textbf{claim} over school $s$ with respect to priority $>_s$ at matching $\mu$ if: (i) $s\notin \mu(i),$ (ii) $s\in C_i (\mu(i)\cup \{s\})$, and (iii) there is $j\in  \mu(s)$ such that $i>_s j$.
 We say that a matching $\mu$ \textbf{eliminates all claims} with respect to $(>_s)_{s\in S}$ if no teacher has a claim.
A matching $\mu$ is \textbf{non-wasteful (for schools)} if whenever $i\in I$, $s\in S$ with $s\notin \mu(i)$ and $s\in C_i(\mu(i)\cup \{s\})$, we have $|\mu (s)|=q_s.$ A matching $\mu$ is \textbf{individually rational}  if  $\mu(i)\succeq_i^B \emptyset$, for each $i\in I$. Notice that this condition is only required for all teachers, this is because of the nature of the priorities of schools. 
We say that a matching $\mu$ is \textbf{stable} for a static problem if it is individually rational, non-wasteful, and eliminates all claims.

\cite{martinez2008invariance} show that the set of stable matchings remains the same when an agent’s preference relation changes, as long as the two preference relations induce the same choice function. While it is well known that substitutable preferences guarantee the existence of stable matchings, we depart from this framework by directly assuming that firms have path-independent choice functions. This broader assumption also ensures the existence of stable matchings \cite[see][for more details]{chamb2017choice}.

A \textbf{dynamic problem}  is a tuple $\mathcal{P}^t=(I^t, S, C_{I^t}, >_S^t,q_S,\mu^{t-1})$. A matching for a dynamic problem is defined similarly in static problems. This problem differs from a static problem in that it also includes the matching $\mu^{t-1}$ (of the previous period).\footnote{In the dynamic problem $\mathcal{P}^t$, when we refer to $\mu^{t-1}$, we mean the restriction of that matching to the set of agents who are present in period $t$.}   Given that in period $t$, $I^t, S, q_S$ and $\mu^{t-1}$ are fixed, some times to ease notation, we simply denote problem  $\mathcal{P}^t$ as $(C_{I^t}, >_S^t)$.
The set of matchings at period $t$ is denoted by $\mathcal{M}^t$.

Given a dynamic problem $\mathcal{P}^t$, we call $I^{t}_E$ to the set of all teachers at period $t$ who were matched at $\mu^{t-1}$ with some school, formally, $I^{t}_E=\{i\in I^{t}: \mu^{t-1}(i)\neq \emptyset\}.$\footnote{Notice that, by the nature of the problem, we implicitly assume that $\mu^{t-1}(i)\succeq_i^B \emptyset$ for each $i\in I^t_E.$} 

Given a domain of dynamic problems, a \textbf{mechanism} is a function $\varphi$ that assigns to each problem $\mathcal{P}^t=(C_{I^t}, >_S^t)$ a matching $\varphi (C_{I^t}, >_S^t)\in\mathcal{M}^t$.

In order to adequately define a solution concept that fits well to a dynamic problem, and captures the fact that teachers hold tenured positions, matchings not only have to fulfill the already known notion of individual rationality, but they must also ensure that teachers who remain from one period to the next do not worsen their situation. 
 The following definition formalizes this notion. 
 
\begin{definition}Given a dynamic problem $\mathcal{P}^t=(I^t, S, C_{I^t}, >_S^t,q_S,\mu^{t-1})$, a matching $\mu^t$ is \textbf{dynamically rational} if it is individually rational (at period $t$), and $\mu^t(i)\succeq_i^B \mu^{t-1}(i)  $ for each $i\in I^t_E.$ 
   
\end{definition}


The other notion we need to adapt to dynamic problems is that of claim. Assume that a teacher has a claim over a school at one period. Given that teachers hold tenured positions, we must consider the teachers assigned to the school at the previous period. 
Recall that a teacher has a claim over a school if the matching assigns another teacher with lower priority to this school. If the teacher matched to the school was already matched from the previous period, the claim is unjustified since teachers hold tenured positions. Otherwise, the claim is justified. Formally,

\begin{definition}
Let  $\mathcal{P}^t=(C_{I^t}, >_S^t)$ be a dynamic problem and let $\mu^t\in \mathcal{M}^t$. Assume that  teacher $i$ has a claim over  school $s$ with respect to priority $>_s^t$ at matching $\mu^t$, so there is $j\in I^t$ such that $j\in \mu^t(s),$ and $i>_s^t j$. We say that the claim is  \textbf{justified} if $j\notin \mu^{t-1}(s)$.  Otherwise, if $j\in \mu^{t-1}(s)$, the claim is \textbf{unjustified}.

\end{definition}

We say that a matching $\mu^t$ \textbf{eliminates justified claims} with respect to $>_S^t$ if whenever teacher  $i$ has a claim over school $s$ with respect to priority $>_s^t$ at matching $\mu^t$, the claim is unjustified. To adapt 
 the notion of stability to dynamic problems, we must relax the requirement of eliminating claims, i.e. we will only request that a matching eliminate justified claims.  Unjustified claims will not be realized since teachers hold tenured positions. Finally, we are in a position to define an appropriate solution concept for dynamic problems.

\begin{definition}
   Let  $\mathcal{P}^t=(C_{I^t}, >_S^t)$ be a dynamic problem. A matching $\mu^t$ is \textbf{dynamically stable} if it is dynamically rational, non-wasteful, and eliminates justified claims with respect to $>_S^t$ at $\mu^t$.
\end{definition}
 Denote by $\mathcal{DS}^t$ the set of all dynamically stable matchings of $\mathcal{P}^t$.

\section{The tenure-respecting deferred acceptance mechanism}\label{section TRDA mecanismo}
Once we have defined the appropriate solution concept for dynamic problems, we must try to answer whether such a solution always exists. In this section, we answer this in the affirmative.  We do this by defining a new priority profile for the schools derived from the dynamic problem taking into account that teachers hold tenured positions. Given a school, the teachers who were matched with that school in the previous period are ranked above all other teachers present in the current period, while maintaining their relative order among themselves. With this new priority, we define a related static problem. Then, we compute a stable matching for this related static problem through the teacher-proposing deferred acceptance mechanism.
We will further show that this stable matching turns out to be dynamically stable in the original dynamic problem. Now we formally define the derived priorities.
\begin{definition}\label{defino violo prioridades}
    Let  $\mathcal{P}^t=(I^t, S, C_{I^t}, >_S^t,q_S,\mu^{t-1})$ be a dynamic problem. For each school $s\in S$ and each priority $>^t_s$,  define \textbf{derived priority $\boldsymbol{\gg_s^t}$} as follows. For each pair $i,j\in I^t$,
    \begin{enumerate}[(i)]
        \item $i\in \mu^{t-1}(s)$ and $j\in I^t\setminus \mu^{t-1}(s)$ imply that $i\gg_s j,$ 
        \item $i,j\in \mu^{t-1}(s)$ implies that  $i\gg_sj\text{ if and only if } i>_s^t j,$  and
        \item $i,j\in I^t\setminus\mu^{t-1}(s)$ implies that  $i\gg_sj\text{ if and only if } i>_s^t j.$  
    \end{enumerate}
\end{definition}

Given a  dynamic problem $\mathcal{P}^t=(I^t, S, C_{I^t}, >_S^t,q_S,\mu^{t-1})$ and  the derived priorities $\gg^t_S$ we can define the \emph{related (static) problem} $$\mathcal{P}=(I^t, S, C_{I^t}, \gg_S^t,q_S).$$
As mentioned earlier, for this related problem, we can compute a stable matching using the deferred acceptance algorithm. Consequently, we can define a mechanism for addressing dynamic problems by considering the related static problem. In this manner, the \textbf{Tenure-Respecting Deferred Acceptance ($\boldsymbol{TRDA}$) mechanism} assigns to each dynamic problem the result obtained by applying the teacher-proposing deferred acceptance mechanism to its corresponding static problem. We present in Table \ref{TRDA mechanism} the $TRDA$ mechanism for dynamic problems.
\medskip
\begin{table}[h!]
    \centering    
    \begin{tabular}{l p{13cm}}\hline
     \textbf{Input:}  &A dynamic  problem $(I^t, S, C_{I^t}, >_S^t,q_S,\mu^{t-1})$. \\
       \textbf{Output:}  &A  dynamically stable matching.\vspace{10 pt}\\
   \textbf{Step $\boldsymbol{1}$}     &Each teacher $i$ proposes to her chosen set of schools $C_i(S)$ and each school $s$ accepts the top $q_s$ teachers among those who have applied to it \emph{with respect to its derived priority} $\gg_s$, and rejects the rest of the proposing teachers. \\
      \textbf{Step $\boldsymbol{k}$}   &Each teacher $i$ proposes to her chosen set of schools that includes all of those schools that she previously proposed to and which have not yet rejected her. Each school $s$ maintains the top $q_s$ teachers among those that propose now and those that were accepted in the previous round, \emph{with respect to its derived priority}, and rejects the rest. The mechanism ends when there are no rejections. \\   \hline
    \end{tabular}
    \caption{The $TRDA$ mechanism}
    \label{TRDA mechanism}
\end{table}

The following theorem states that our mechanisms for dynamic problems generate a dynamically stable matching. Formally,

\begin{theorem}\label{conjunto de estables dinamico no vacio}
For the domain of dynamic problems in which teachers have path-independece choice functions, the $TRDA$ mechanism is dynamically stable. 
\end{theorem}
\begin{proof}
    See Appendix \ref{apendice de pruebas}.
\end{proof}

A desirable property of a mechanism is efficiency. This efficiency can be considered either within or outside the class of dynamically stable matchings. In this section, we will study efficiency within the class of dynamically stable matchings. We say that a matching $\overline{\mu}$ Blair-dominates a matching $\mu$ if every teacher weakly Blair-prefers $\overline{\mu}$ to $\mu$, and at least one teacher strictly Blair-prefers $\overline{\mu}$ to $\mu$. Formally, $\overline{\mu}$  \textbf{Blair-dominates} $\mu$ if $\overline{\mu}(i) \succeq^B_i \mu(i)$ for each $i \in I^t$, and there is $j \in I^t$ such that $\overline{\mu}(j) \succ^B_j \mu(j)$. A mechanism $\overline{\varphi}$ Blair-dominates mechanism $\varphi$ if $\overline{\varphi}\neq \varphi$ and  $\overline{\varphi}(C_{I^t}, >_S^t) \succeq^B_{i} \varphi(C_{I^t}, >_S^t)$ for each $(C_{I^t}, >_S^t)$ and each $i\in I^t$. 
 Let $\mu_{DA}$ denote the outcome of the teacher-proposing deferred acceptance mechanism. As stated by \cite{Roth1984b}, in a static setting, $\mu_{DA}\succeq_i^B\mu$ for any other stable matching $\mu$ and each $i\in I.$ For the dynamic setting, let $\mu_{TR}$ denote the outcome of the $TRDA$ mechanism. Therefore, as a consequence of Theorem \ref{conjunto de estables dinamico no vacio}, we have that $\mu_{TR}\succeq_i^B\mu$ or any other dynamically stable matching $\mu$ and each $i\in I.$  In this case, we say that the outcome of the $TRDA$ is \textbf{(Blair) constrained-efficient} (within the class of dynamically stable matchings). Thus, we can formally state the following theorem.

\begin{theorem}\label{teorema pareto superioridad del mecanismo}
For the domain of dynamic problems in which teachers have path-independece choice functions, the outcome of the $TRDA$ is constrained-efficient.
\end{theorem}

Note that, given a dynamic problem $\mathcal{P}^t=(C_{I^t}, >_S^t),$ a dynamically stable matching may have unjustified claims, but not all dynamically stable matchings have the same amount of unjustified claims. Thus, given a dynamically stable matching $ \mu^t$, let $\mathcal{U}(\mu^t)$  be the set of unjustified claims with respect to $>_S^t$ at $\mu^t.$ We can classify the dynamically stable matchings with respect to the amount of unjustified claims. We say that a matching $\mu^t$ \textbf{minimizes unjustified claims} if (i) it is dynamically stable, and (ii) there is no dynamically stable matching $\overline{\mu}^t$ such that $\mathcal{U}(\overline{\mu}^t)\subsetneq \mathcal{U}(\mu^t)$.
     Notice that for a dynamic problem, there is always a matching that minimizes unjustified claims. This is because the set $\mathcal{DS}^{t}$ is finite and, by Theorem \ref{conjunto de estables dinamico no vacio}, non-empty. The following theorem presents an important property of the $TRDA$ mechanism.

\begin{theorem}\label{TRDA minumica reclamos injustificados}
  For the domain of dynamic problems in which teachers have path-independece choice functions, the outcome of the $TRDA$ mechanism minimizes unjustified claims.
\end{theorem}
\begin{proof}
   See Appendix \ref{apendice de pruebas}.
\end{proof}


\section{Improving efficiency beyond the class of dynamically stable mechanisms}\label{seccion de eficiencia}

In the previous section, we prove that the $TRDA$ mechanism is constrained-efficient. If we allow matchings to have justified claims, i.e., matchings that are not dynamically stable, we can enhance the efficiency of matchings. In this section, we explore a mechanism whose outcome Blair-dominates the outcome of the $TRDA$ mechanism (when those outcomes are different).
It is well established that in static problems, the teacher-proposing deferred acceptance mechanism may leave room for improving efficiency beyond the stable matching class. Consequently, the $TRDA$ mechanism also leaves room for improving efficiency beyond the class of dynamically stable mechanisms. \cite{kesten2010school} proposes an adaptation of the $DA$ mechanism that yields a matching that is more efficient than the one produced by the teacher-proposing deferred acceptance mechanism, at the expense of losing stability. Such an improving mechanism is the Efficiency Adjusted Deferred Acceptance ($EADA$) mechanism. This mechanism may fail to produce a stable matching, as it seeks to improve efficiency by requiring teachers to consent to potential violations of their priorities. As a result, the outcome may be subject to claims.

In this section, we adapt the mechanism proposed by \cite{kesten2010school} for our dynamic many-to-many school choice model. 
This adaptation must take into account two structural aspects of the model: its many-to-many nature and its dynamic evolution over time. These dimensions are reflected in two key features of teachers—their tenured positions and their path-independent choice functions.  We refer to this modified mechanism as the \textbf{Tenure-Respecting Efficiency Adjusted Deferred Acceptance ($\boldsymbol{TREADA}$) mechanism}. So, in dynamic many-to-many problems where teachers have tenured positions, the $TREADA$ mechanism works as follows.
First, the mechanism determines a dynamically stable matching using the $TRDA$ mechanism and subsequently identifies interrupter pairs. According to \cite{kesten2010school}, a pair $(i, s)$ is called an \textbf{interrupter pair} at step $p'$ if teacher $i$ is tentatively assigned to school $s$ at some step $p < p'$, and is later rejected by school $s$ at step $p'$, \emph{provided that} at least one other teacher is rejected by $s$ at some step $l \in \{p, p+1, \dots, p'-1\}$.\footnote{In this case, we also say that teacher~$i$ is said to be an interrupter for school~$s$ at step~$p'$.} Second, for each interrupter pair $(i, s)$ in which teacher~$i$ consents to the violation of her priorities---in which case $i$ is called a \textbf{consenting teacher}---the \textsc{TREADA} mechanism modifies the choice function of teacher~$i$ by excluding all sets that contain school~$s$.
 The choice functions of other teachers remain unchanged. Once this modification is made, the $TRDA$ is rerun. Interrupter pairs are identified again, a new restriction of each consenting interrupter teacher choice function is performed, and a new round is conducted. This process continues until no interrupter pairs remain, and when all teachers consent it results in an \textbf{efficient} matching, i.e., a matching that is not Blair-dominated by any other matching. 


Therefore, before formally presenting the $TREADA$ mechanism, it is important to note that, in our school choice model where teachers have path-independent choice functions, given a pair $(i, s)$, we need to modify the choice function of teacher~$i$ in order to always selects subsets of schools not containing school $s$. Formally,
\begin{definition}\label{defino s-truncacion}
    Given a pair $(i,s)$ and a choice function $C_i$, the \textbf{$\boldsymbol{s$-truncation of $C_i}$} is the choice function $C^s_i: 2^S \to 2^{S\setminus\{s\}}$ such that
    $$C^s(S')=C(S'\setminus \{s\})$$ for each $S'\subseteq S.$
 \end{definition}   
   
This definition generalizes the preference truncation process introduced by \cite{Martinez2004} and further employed by \cite{bonifacio2022cycles}. Here, we prove that when a choice function is path-independent, the $s$-truncation inherits such a property. Formally,

\begin{lemma}\label{lema s-truncacion de C_i}
 Given a pair $(i,s)$ and a path-independent choice function $C_i$, the $s$-truncation $C^s_i$ is also path-independent.    
\end{lemma}
\begin{proof}
  See Appendix \ref{apendice de pruebas}.   
\end{proof}

Now, we formally present the $TREADA$ mechanism in Table \ref{TREADA mechanism}.
\medskip
\begin{table}[h!]
    \centering    
    \begin{tabular}{l p{13cm}}\hline
     \textbf{Input:}  &A dynamic  problem $(I^t, S, C_{I^t}, >_S^t,q_S,\mu^{t-1})$. \\
       \textbf{Output:}  &A matching that Blair-dominates the output of the $TRDA$ mechanism.\vspace{10 pt}\\
  \textbf{Round 0}&Execute the $TRDA$ mechanism.\\
       \textbf{Round $\mathbf{u}$}&
       Identify the last step of the  $TRDA$ mechanism executed in round $u-1$ where a consenting teacher is rejected by the school for which he is an interrupter. Identify all interrupter pairs from that step involving a consenting teacher. If there are no interrupter pairs, stop. For each identified interrupter pair $(i, s)$, perform the $s$-truncation of $C_i$. Do not change other agents' choice functions. Call the new profile of choice functions for teachers $C^{u}_{I^t}$. Rerun the $TRDA$ mechanism for the dynamic problem $(I^t, S,C^{u}_{I^t}, >_S^t,q_S,\mu^{t-1})$.\\ \hline
    \end{tabular}
    \caption{The $TREADA$ mechanism}
    \label{TREADA mechanism}
\end{table}


The following example illustrates the application of the $TREADA$ mechanism to solve a dynamic problem where teachers have tenured positions and path-independent choice functions.
We show how interrupter pairs are identified and managed through truncation of choice functions, showing how eliminating these pairs affects the final matching outcome improving its efficiency.

\begin{example}\label{ejemplo del TREADA}
    Consider the following dynamic problem $(I^t, S, C_{I^t}, >_S^t,q_S,\mu^{t-1})$, where \\$I^t=\{i_1,i_2,i_3,i_4\}$, $S=\{s_1,s_2,s_3,s_4\}$, $q_{s_1}=2$ and $q_{s_\ell}=1$ for $\ell=2,3,4$ and $$\mu^{t-1}=\left(\begin{matrix}
        s_2 & s_3 & s_4\\
        i_1 & i_3 &i_2
    \end{matrix}\right).
    $$ Teachers' choice functions and schools' priorities are given in the following tables:
    
    \begin{center}
    \begin{tabular}{c|cccc}
       $2^{I^t}$  & $C_{i_1}$ & $C_{i_2}$ & $C_{i_3}$ & $C_{i_4}$\\ \hline
  $\{s_1,s_2,s_3,s_4\}$ & $\{s_4\}$ & $\{s_2,s_3\}$ & $\{s_4\}$ & $\{s_2,s_3\}$\\
  $\{s_1,s_2,s_3\}$     & $\{s_2\}$     & $\{s_2,s_3\}$     & $\{s_2\}$ & $\{s_2,s_3\}$\\
  $\{s_1,s_2,s_4\}$     & $\{s_4\}$     & $\{s_2\}$     & $\{s_4\}$     & $\{s_2\}$\\
  $\{s_1,s_3,s_4\}$     & $\{s_4\}$     & $\{s_3\}$     & $\{s_4\}$     & $\{s_3\}$\\
  $\{s_2,s_3,s_4\}$     & $\{s_4\}$         & $\{s_2,s_3\}$         & $\{s_4\}$     & $\{s_2,s_3\}$\\
  $\{s_1,s_2\}$         & $\{s_2\}$         & $\{s_2\}$         &  $\{s_2\}$            &  $\{s_2\}$          \\
  $\{s_1,s_3\}$         & $\{s_3\}$         & $\{s_3\}$         &  $\{s_3\}$            &   $\{s_3\}$         \\
  $\{s_1,s_4\}$         &  $\{s_4\}$                &   $\{s_4\}$               &     $\{s_4\}$         &  $\{s_4\}$          \\
  $\{s_2,s_3\}$         &    $\{s_2\}$    & $\{s_2,s_3\}$   &     $\{s_2\}$   &     $\{s_2,s_3\}$          \\
  $\{s_2,s_4\}$         &    $\{s_4\}$   &   $\{s_2\}$  & $\{s_4\}$     &   $\{s_2\}$      \\
  $\{s_3,s_4\}$         &     $\{s_4\}$      &    $\{s_3\}$     &    $\{s_4\}$  &  $\{s_3\}$ \\
  $\{s_1\}$             &      $\{s_1\}$       &    $\{s_1\}$      &   $\{s_1\}$    & $\{s_1\}$   \\
  $\{s_2\}$             &   $\{s_2\}$    &           $\{s_2\}$   &     $\{s_2\}$      & $\{s_2\}$   \\
  $\{s_3\}$             &    $\{s_3\}$   &   $\{s_3\}$   &   $\{s_3\}$   &  $\{s_3\}$    \\
  $\{s_4\}$             &     $\{s_4\}$   &      $\{s_4\}$  &   $\{s_4\}$  & $\{s_4\}$ \\
\end{tabular}\hspace{30pt}
            \begin{tabular}{cccc}
       $>_{s_1}^t$ & $>_{s_2}^t$  & $>_{s_3}^t$ & $>_{s_4}^t$ \\ \hline
       $i_4$  & $i_3$  & $i_1$ & $i_4$  \\
       $i_1$  & $i_2$ & $i_3$ & $i_1$\\
       $i_3$  & $i_1$ &  $i_2$ &$i_3$ \\
       $i_2$  & $i_4$ & $i_4$ & $i_2$\\
       &&&\\
        &&&\\
        &&&\\
        &&&\\
        &&&\\
        &&&\\&&&\\
        &&&\\
        &&&\\
        &&&\\
        &&&\\
        \end{tabular}
   \medskip
   \end{center}
The derived priorities that give rise to the related static problem in order to run the $TRDA$ mechanism are given in the following table 
    \begin{center}
        \begin{tabular}{cccc}
         $\gg_{s_1}^t$ & $\gg_{s_2}^t$  & $\gg_{s_3}^t$ & $\gg_{s_4}^t$ \\ \hline
       $i_4$  & $i_1$  & $i_3$ & $i_2$  \\
       $i_1$  & $i_3$ & $i_1$ & $i_4$\\
       $i_3$  & $i_2$ &  $i_2$ &$i_1$ \\
       $i_2$  & $i_4$ & $i_4$ & $i_3$\\
        
        \end{tabular}
    \end{center}

\medskip
\noindent Now, we apply the $TREADA$ mechanism. 

\noindent \textbf{Round $\mathbf{0}$:} The $TRDA$ mechanism is applied. 
\begin{center}

 \begin{tabular}{ccccc}
      & $s_1$ & $s_2$  & $s_3$& $s_4$  \\ \hline
      {Step 1} &  & $i_2$ \bcancel{$i_4$} & $i_2$ \bcancel{$i_4$} & $i_1$ \bcancel{$i_3$} \\
{Step 2} &  & \bcancel{$i_2$} $i_3$ & $i_2$ & \bcancel{$i_1$} $i_4$ \\
{Step 3} &  & \bcancel{$i_3$} $i_1$ & $i_2$ & $i_4$ \\ 
{Step 4} &  & $i_1$ & \bcancel{$i_2$} $i_3$ & $i_4$ \\
{Step 5} &  & $i_1$ & $i_3$ & \bcancel{$i_4$} $i_2$ \\
{Step 6} & $i_4$ & $i_1$ & $i_3$ & $i_2$ \\   
    \end{tabular}
  
\end{center}

\noindent Thus, the output of Round $0$ is the matching $$ \mu^0 = 
    \begin{pmatrix}
    i_1 & i_2 & i_3 & i_4 \\
    s_2 & s_4 & s_3 & s_1
    \end{pmatrix}$$

\noindent \textbf{Round $\mathbf{1}$:} The pair $(i_4, s_4)$ is an interrupter within the algorithm, as teachers $i_1$ was rejected while teacher $i_4$ was tentatively placed in school $s_4$. Then, consider the  $s_4$-truncation of $C_{i_4}$.
The execution of the corresponding $TRDA$ mechanism for the revised problem is presented in the following table.

\begin{center}
   
    \begin{tabular}{ccccc}
      & $s_1$ & $s_2$  & $s_3$ & $s_4$ \\ \hline
    {Step 1} & {$i_4$} & {$i_2$ \bcancel{$i_4$}} & {$i_2$ \bcancel{$i_4$}} & {$i_1$ \bcancel{$i_3$}} \\
    {Step 2} & {$i_4$} & {\bcancel{$i_2$} $i_3$} & {$i_2$} & {$i_1$} \\
    {Step 3} & {$i_4$} & {$i_3$} & {$i_2$} & {$i_1$} \\
    \end{tabular}
    
\end{center}
\noindent Thus, the output of Round $1$ is  matching $$\mu^1=\left(\begin{matrix} i_1 & i_2 & i_3 & i_4 \\ 
s_4 & s_3 & s_2 & s_1 \end{matrix}\right).$$ 

\noindent \textbf{Round $\mathbf{2}$:} The pair $(i_2, s_2)$ is an interrupter since teacher $i_4$ was rejected while teacher $i_2$ was tentatively placed in school $s_2$. Then, consider the revised problem where we perform the $s_2$-truncation of $C_{i_2}$.
The execution of the corresponding $TRDA$ mechanism for the revised problem is presented in the following table.
\begin{center}
        \begin{tabular}{ccccc}
      & $s_1$ & $s_2$  & $s_3$ & $s_4$ \\ \hline
        {Step 1} & { } & {$i_4$} & {$i_2$ \bcancel{$i_4$}} & {$i_1$ \bcancel{$i_3$}} \\
    {Step 2} & { } & {\bcancel{$i_4$} $i_3$} & {$i_2$} & {$i_1$} \\
    {Step 3} & {$i_4$} & {$i_3$} & {$i_2$} & {$i_1$} \\
    \end{tabular}
   
\end{center}
Note that, there are no interrupter pairs in this round, thus the mechanism stops obtaining matching
$$\mu^2 = \left(\begin{matrix} i_1 & i_2 & i_3 & i_4 \\ s_4 & s_3 & s_2 & s_1 \end{matrix}\right).$$ \hfill$\Diamond$

\end{example}



This way, by applying the $TREADA$ mechanism for a dynamic problem, we can obtain a matching that Blair-dominates the outcome of the $TRDA$ mechanism. The extent of the improvement will depend on the consent of the teachers. So, the efficiency, beyond the class of the dynamically stable matchings, will be reached when all teachers consent. Formally,

\begin{theorem}\label{teorema TREADA eficiente}
The $TREADA$ mechanism Blair-dominates the $TRDA$ mechanism as well as any other dynamically stable mechanism. When no teacher consents, $TREADA$ and $TRDA$ mechanisms yield the same outcomes. However, if all teachers consent, the outcome of the $TREADA$ mechanism achieves efficiency. 
\end{theorem}
\begin{proof}
   The proof is presented in Appendix \ref{apendice de pruebas} using Lemmata  \ref{lema 2 del eada} and \ref{lema 3 del eada}.
\end{proof}

\section{Not obviously manipulability in dynamic problems}\label{Seccion de manipulation no obvia}

In this section, we analyze the potential for manipulation by teachers under the mechanisms introduced in this paper. While the mechanisms were initially defined assuming that teachers report path-independent choice functions, we now consider a setting in which teachers instead report preferences over subsets of schools. This alternative formulation is particularly useful for studying strategic behavior, as it allows us to frame manipulation in terms of misreported preferences.

Importantly, this change in the domain of reports does not affect the applicability of our dynamic analysis. The results established in the previous sections remain valid when teachers report substitutable preferences instead of path-independent choice functions. This equivalence is justified by well-established connections between the two domains. Specifically, a choice function is path-independent if and only if it is both substitutable and consistent \citep[see][for more details]{aizerman1981general,chamb2017choice}. Furthermore, any substitutable preference relation induces a choice function that satisfies both substitutability and consistency \citep[see][for more details]{martinez2008invariance}. Therefore, restricting attention to substitutable preferences preserves the essential structure of the problem while providing a more natural framework for analyzing manipulation.
\footnote{As observed in \cite{martinez2008invariance}, a preference relation induces a choice function in a natural way. Given a preference $\succ_i$ and a subset $S' \subseteq S$, the corresponding choice function is defined as $C_i(S') = \max_{\succ_i} { S'' \subseteq S' }$, i.e., the most preferred subset of $S'$ according to $\succ_i$.}

Since we work directly with preferences over sets of schools, the relevant notion of dominance in this setting is that of Pareto-domination, rather than Blair-domination over choice functions that we previously used.\footnote{A matching $\overline{\mu}$ \textbf{Pareto-dominates} $\mu$ if $\overline{\mu}(i) \succeq_i \mu(i)$ for each $i \in I^t$, and there is $j \in I^t$ such that $\overline{\mu}(j) \succ_j \mu(j)$. A mechanism $\overline{\varphi}$ Pareto-dominates a mechanism $\varphi$ if $\overline{\varphi} \neq \varphi$ and $\overline{\varphi}(\succ_{I^t}, >S^t) \succeq_{i} \varphi(\succ_{I^t}, >S^t)$ for each $(\succ{I^t}, >_S^t)$ and each $i \in I^t$.} Pareto-domination allows us to evaluate outcomes based on teachers' reported preferences in a way that is both intuitive and compatible with the strategic framework we consider. This shift in focus from Blair-domination to Pareto-domination reflects our preference-based formulation and facilitates a more natural analysis of the incentive properties of the mechanisms.

In the context of static many-to-many matching problems, \cite{martinez2004group} show that the substitutability condition on teachers' preferences is insufficient to prevent a teacher from benefiting from misreporting her preferences under the teacher-proposing deferred acceptance mechanism. This finding highlights that, beyond stability and efficiency, the (non-)manipulability of a mechanism plays a central role in the two-sided matching literature.

In our setting, due to the nature of schools' priorities, only teachers can behave strategically. This issue has also been explored in many-to-many contexts by \cite{sakai2011note, hirata2017stable, iwase2022equivalence}, among others. Here, we focus on \emph{dynamically stable-dominating mechanisms}, which are either dynamically stable mechanisms or mechanisms that (teacher) Pareto-dominate a dynamically stable one.

A teacher manipulates a mechanism if there is a scenario in which she can achieve a better outcome by declaring preferences that differ from her true ones. As mentioned in the Introduction, it is well-established that in the many-to-one matching model with substitutable preferences, any stable mechanism may be vulnerable to manipulation. Consequently, stable mechanisms are also susceptible to manipulation in our many-to-many static context and, by extension, in our dynamic many-to-many context. Since manipulation cannot be fully eliminated, we focus on dynamically stable-dominating mechanisms that at least prevent obvious manipulations, as defined by \cite{troyan2020obvious}.

A manipulation is considered \emph{obvious} if agents can easily recognize and execute a beneficial misrepresentation of their true preferences. In models with substitutable preferences, the notion of non-obvious manipulation is studied by \cite{arribillaga2023obvious} for static many-to-one problems and by \cite{arribillaga2024many-yo-many-obvious} for static many-to-many problems.



To describe the term ``obvious" in static problems, it is necessary to specify how much information each agent has about other agents' preferences. \cite{troyan2020obvious} assume that each agent has complete ignorance in this respect and therefore focus on all outcomes that the mechanism can choose given its own report. A manipulation is considered obvious if the best possible outcome under the manipulation is strictly better than the best possible outcome under truth-telling, or if the worst possible outcome under the manipulation is strictly better than the worst possible outcome under truth-telling. \cite{troyan2020obvious} prove that any stable-dominating mechanism is not obviously manipulable in a many-to-one setting.


 A mechanism $\varphi$ is said to be not obviously manipulable if it does not admit obvious manipulations.
 In a many-to-many static problem with substitutable preferences, \cite{arribillaga2024many-yo-many-obvious} proved that any stable-dominating mechanism is not obviously manipulable. Formally, 

 \begin{theorem}[\citealp{arribillaga2024many-yo-many-obvious}]\label{teorema 2 arribi}Facts about not obvious manipulability:

 \begin{enumerate}[(i)]
     \item The teacher-proposing deferred acceptance mechanism is not obviously manipulable.
     \item Any mechanism that Pareto-dominates the teacher-proposing deferred acceptance mechanism is not obviously manipulable.
 \end{enumerate}
   \end{theorem}

If we turn our attention back to the dynamic setting, in order to formally define the notions of dynamic manipulation and obvious dynamic manipulation, we must first introduce the notion of an economy, which is essentially a sequence of dynamic problems. Formally, an \textbf{economy}  $$\mathcal{E}=(\{I^t\}_t, S, \{\succ_{I^t}\}_t, \{>_S^t\}_t,q_S,\mu^0,\varphi)$$ is defined by a sequence of sets of teachers $\{I^t\}_t$, the set of schools $S$, a sequence of preference profiles for the set of teachers $\{\succ_{I^t}\}_t$, a sequence of priorities profiles for the set of schools $\{>_S^t\}_t$, a  profile of quotas $q_S$,  an initial matching $\mu^{0}$, and a mechanism $\varphi.$ This mechanism generates a matching at each period. We assume that schools’ priorities over teachers remain fixed over time, that is, for each $s \in S$ and any two periods $t$ and $t'$, $i >_s^t j$ if and only if $i >_s^{t'} j$ for all $i, j \in I^t \cap I^{t'}$. Moreover, teachers report their preferences upon entering the economy and do not revise them over time.\footnote{Although this is a strong assumption, in Section \ref{seccion de discusion} we discuss the case in which teachers are allowed to revise their preferences over time.}

 To carry out the analysis of manipulability of the mechanisms, we introduce the domain of substitutable preferences, which we denote by $\mathcal{D}$.
 
Now, for a dynamic context, let $\varphi$ be a mechanism, and let $i \in I^t$ be an agent that enters the problem at period $t$ with true preference $\succ_i\in \mathcal{D}$. We say that  $\succ'_i\in \mathcal{D}$  is a \textbf{dynamic manipulation of $\boldsymbol{\varphi$ at $\succ_i}$} if in some period $t'\geq t$ we have that
$$\varphi_i (\succ'_i, \succ_{I^{t'}\setminus\{i\}},>_S^{t'} )\succ_i\varphi_i(\succ_i, \succ_{I^{t'}\setminus\{i\}},>_S^{t'} ).$$ 
To formally define when a manipulation is obvious in a dynamic problem, we need to define the \textbf{option set left open by 
$\boldsymbol{\succ_i$ at $\varphi$ in period $t}$} as
\[
O^t_\varphi(\succ_i) = \{\varphi_i(\succ_i, \succ_{I^{t}\setminus\{i\}}, >_S^t)~ :~ \succ_{I^{t}\setminus\{i\}} \in  \mathcal{D}^{|I^t\setminus\{i\}|}  \}.
\]
The following definition adapts the concept of obvious manipulation to our dynamic context.

\begin{definition}\label{defino nom}
Let $\varphi$ be a mechanism, and let $i \in I^t$ be an agent that enters the problem at period $t$ with true preference $\succ_i\in  \mathcal{D}$. Let $\succ'_i\in  \mathcal{D}$ be a dynamic manipulation of $\varphi$ at $\succ_i$ in some period $t'\geq t$.
Then, $\succ'_i$ is a \textbf{obvious dynamic manipulation} if
\[
W_i(O^{t'}_\varphi(\succ'_i)) \succ_i W_i(O^{t'}_\varphi(\succ_i))
\]
or
\[
C_i(O^{t'}_\varphi(\succ'_i)) \succ_i C_i(O^{t'}_\varphi(\succ_i)).
\]
A mechanism $\varphi$ is \textbf{not obviously dynamically manipulable} if it does not admit any obvious dynamic manipulation. Otherwise, $\varphi$ is \textbf{obviously dynamically manipulable}.
\end{definition} 


The following example shows that, in dynamic problems with substitutable preferences
, the $TRDA$  mechanism has an obvious dynamic manipulation.

\begin{example}
    Consider a two-period economy where   
$\mu^{0}{=\left(\begin{matrix}
 i_1 & i_2  &  i_3  \\
s_2 &s_4 & s_3 
\end{matrix}\right)},$ 
$I^1=\{i_1,i_2,i_3,i_4\}$ (and thus $I_E^1=\{i_1,i_2,i_3\}$), $S=\{s_1,s_2,s_3,s_4\}$, $q_1=2$, $q_\ell=1$ for $\ell=2,3,4$, and the mechanism considered is $TRDA$. 
Teachers' preferences and schools' priorities  at period $1$ are given by: 

\begin{center}
    \begin{tabular}{cccc}
       $\succ_{i_1}$  & $\succ_{i_2}$ & $\succ_{i_3}$ & $\succ_{i_4}$ \\ \hline
       $s_4$ & $s_2s_3$  &  $s_4$ &$s_2s_3$\\
       $s_2$& $s_2$    &   $s_2$ &$s_2$ \\
       $s_3$ & $s_3$   &  $s_3$ &$s_3$ \\
      $s_1$ &   $s_4$  & $s_1$ &$s_4$\\
      &   $s_1$         &    &  $s_1$
    
    \end{tabular}
\hspace{50pt}
          \begin{tabular}{cccc}
       $>_{s_1}^1$ & $>_{s_2}^1$  & $>_{s_3}^1$ & $>_{s_4}^1$ \\ \hline
      $i_1$  & $i_1$  & $i_1$ & $i_2$  \\
       $i_2$  & $i_3$ & $i_3$ & $i_4$\\
       $i_3$  & $i_2$ &  $i_2$ &$i_1$ \\
       $i_4$  & $i_4$ & $i_4$ & $i_3$\\

       &&&
    \end{tabular}
      \end{center} 
\medskip

\noindent Then, the outcome of the $TRDA$ mechanism at period $1$ is: 
$$\mu^1{=\left(\begin{matrix}
 i_1 & i_2  & i_3 & i_4  \\
s_2 & s_4 & s_3 & s_1 
\end{matrix}\right)}.$$
Let us assume that $I^{2}=\{i_1,i_3,i_4,i_5\}$ and teachers' preferences and schools' priority at period $2$ are given by:

 \begin{center}
    \begin{tabular}{cccc}
       $\succ_{i_1}$  & $\succ_{i_3}$ & $\succ_{i_4}$ & $\succ_{i_5}$ \\ \hline
       $s_4$ &   $s_4$ &$s_2s_3$&$s_2s_3$ \\
       $s_2$&   $s_2$ &$s_2$ &$s_2$    \\
       $s_3$ &   $s_3$ &$s_3$&$s_3$   \\
      $s_1$   & $s_1$ &$s_4$&   $s_4$\\
      &          &  $s_1$  &  $s_1$
    
    \end{tabular}
\hspace{50pt}
             \begin{tabular}{cccc}
       $>_{s_1}^{2}$ & $>_{s_2}^{2}$  & $>_{s_3}^{2}$ & $>_{s_4}^{2}$ \\ \hline
      $i_1$  & $i_1$  & $i_1$ & $i_5$  \\
       $i_3$  & $i_3$ & $i_3$ & $i_4$\\
       $i_5$  & $i_4$ &  $i_4$ &$i_1$ \\
       $i_4$  & $i_5$ & $i_5$ & $i_3$\\
       &&&
    \end{tabular}
         \end{center} 
 
\noindent Then, the outcome of the $TRDA$ mechanism at period $2$ is: 
$$\mu^{2}{=\left(\begin{matrix}
 i_1 & i_3  & i_4 & i_5  \\
s_2 & s_3 & s_1 & s_4 
\end{matrix}\right)}.$$
Assume that teacher $i_4$ reveals preference $\succ_{i_4}':s_2s_3, s_2, s_3, s_1, s_4$ instead of revealing her true preference.\footnote{$\succ_{i_4}':s_2s_3, s_2, s_3,  s_1, s_4$ indicates that $  s_2s_3\succ_{i_4}' s_2,$  $s_2\succ_{i_4}' s_3,$ and so-forth.}  Then, the matchings generated by the $TRDA$ mechanism in periods $1$ and $2$ are: 
$$\overline{\mu}^{1}{=\left(\begin{matrix}
 i_1 & i_2  & i_3 & i_4  \\
s_4 & s_3 & s_2 & s_1 
\end{matrix}\right)}\hspace{20pt} \text{ and }\hspace{20pt}
\overline{\mu}^{2}{=\left(\begin{matrix}
 i_1 & i_3 & i_4 & i_5  \\
s_4 & s_2 & s_3 & s_1 
\end{matrix}\right)}.$$
Since $\overline{\mu}^{2}(i_4)=s_3\succ_{i_4}s_1=\mu^{2}(i_4)$, then teacher $i_4$ can benefit by misrepresenting her preferences. 
To see that $\succ'_{i_4}$ is an obvious manipulation, it is enough to show that the worst outcome that teacher $i_4$ can obtain by misrepresenting her preference is preferred, according to their true preference, than the worst outcome that teacher $i_4$ can obtain when she reveals her true preference. Suppose that teacher $i_4$  declares misrepresenting preference $\succ'_{i_4}$, while the other teachers declare the following preferences when entering the problem:
\begin{center} 
\begin{tabular}{cccc} $\succ_{i_1}$ & $\succ_{i_2}$ & $\succ_{i_3}$ & $\succ_{i_5}$ \\ \hline $s_1$ & $s_1$ & $s_2,s_3$ & $s_1$ \\ \vdots & \vdots & \vdots & \vdots  \end{tabular}
\end{center}
Then, the matchings generated in periods $1$ and $2$ are:
 $$\overline{\mu}^{1}{=\left(\begin{matrix}
 i_1 & i_2  & i_3 & i_4  \\
s_1 & s_1 & s_2s_3 & s_4 
\end{matrix}\right)}\hspace{20pt} \text{ and }\hspace{20pt}
\overline{\mu}^{2}{=\left(\begin{matrix}
 i_1 & i_3 & i_4 & i_5  \\
s_1 & s_2,s_3 & s_4 & s_1 
\end{matrix}\right)}.$$
Therefore, the worst case for teacher $i_4$ under the misrepresenting preference $\succ'_{i_4}$ is to work at school $s_4$ while the worst case possible when teacher $i_4$ reveals her true preference $\succ_{i_4}$ is to work at school $s_1$. Since $\{s_4\}\succ_{i_4}\{s_1\}$, then $\succ'_{i_4}$ is an obvious manipulation. Therefore,  the $TRDA$ mechanism is obviously dynamically manipulable.\hfill $\Diamond$
\end{example}

An additional property must be imposed on the set of priorities to restore the condition of non-obvious dynamic manipulability in the mechanism. Specifically, suppose each priority ensures that in every period, teachers matched in the previous period are prioritized over those entering the problem in the current period. In that case, it can be shown that the $TRDA$ mechanism becomes non-obviously dynamically manipulable. It is important to note that this property requires not only that teachers matched to a school in the previous period maintain their tenured positions at that school, but also that they are given priority over newly entering teachers at all schools. This additional property for schools' priorities is formalized as follows:

\begin{definition}
    A set of priorities $\{>_s^t\}_{s\in S}$ is \textbf{lexicographic by tenure} if for all teachers $i,j\in I^t$, whenever $i\in I_E^t$ and $j\notin I_E^t$ we have $i>_s^t j$ for each school $s\in S$.
\end{definition}
  In the following theorem, we state that when schools' priorities satisfy the previously defined property, although the $TREADA$ and all dynamically stable-dominating mechanisms (in particular, the $TRDA$ mechanism) are dynamically manipulable, they are not obviously so.
 
 \begin{theorem}\label{teorema TREADA es NOM}
For any economy in which teachers' preferences satisfy substitutability, schools' priorities are lexicographic by tenure, and the mechanism considered is the $TREADA$ or any dynamically stable-dominating one, such a mechanism is not obviously dynamically manipulable.
\end{theorem}
\begin{proof}
  See Appendix \ref{apendice de pruebas}.   
\end{proof}

\section{Final remarks}\label{seccion de discusion}

In this paper, we consider a dynamic many-to-many school choice problem where schools have priorities over teachers, and teachers have path-independent choice functions.  We consider a setting where the set of teachers may change over time: new teachers can enter, others may leave in any given period, and those who remain matched across periods are assumed to hold tenured positions.

For this dynamic context, we propose a new concept of dynamic stability, we also demonstrate that the set of dynamically stable matchings is non-empty, and present a mechanism that, for each dynamic problem, computes a dynamically stable matching (the $TRDA$ mechanism). We show that the outcome of the $TRDA$ mechanism is constrained-efficient within the set of dynamically stable matchings and that, although a dynamically stable matching may have unjustified claims, the matching produced by the $TRDA$ mechanism minimizes such unjustified claims. As in static problems, stable mechanisms in dynamic problems leave room for efficiency improvements. We adapt the $EADA$ mechanism introduced by \cite{kesten2010school} to our dynamic many-to-many setting, resulting in the $TREADA$ mechanism. We show that, when teachers consent, the $TREADA$ mechanism produces an efficient matching, which may not be dynamically stable. 
To the best of our knowledge, this paper is the first to extend the $EADA$ mechanism introduced by \citet{kesten2010school} to a many-to-many setting where teachers have path-independent choice functions and schools have priorities.

It is important to note that, unlike in the setting studied by \citet{kesten2010school}—a many-to-one school choice market where efficiency is evaluated using the linear order of students—when assuming that teachers have path-independent choice functions (or substitutable preferences), it is crucial to rely on Blair’s partial order. This was already observed by \citet{Roth1984b}, who used Blair’s partial order to characterize the optimality of the $DA$ mechanism for the side with substitutable preferences, and further developed by \citet{Blair1988} to establish the lattice structure of stable matchings. 

Furthermore, when assuming that teachers declare substitutable preferences upon entering the market, we show that although a teacher may benefit from manipulating both the $TRDA$ and $TREADA$ mechanisms, these mechanisms are non-obviously dynamically manipulable under certain additional conditions on schools' priority structures.

A natural question that arises is whether it is possible to restrict the domain of teachers' preferences in such a way that the mechanism becomes strategy-proof. \cite{romero2021two} demonstrate that in school choice problems, where teachers' preferences are responsive (a more restrictive requirement than substitutability) and schools' priorities satisfy an ``acyclicity'' condition, the teacher-proposing deferred acceptance mechanism is strategy-proof.\footnote{A \textbf{cycle} in schools' priorities occurs when there is an alternate list of schools and teachers presenting a cyclical behavior in the following way: each school prefers the teacher to its right over the teacher to its left, and finds both acceptable. Formally, a cycle (of length $p + 1$) in $>_S$ is given by distinct teachers $i_0, i_1, \dots, i_p \in I$ and distinct schools $s_0, s_1, \dots, s_p \in S$, such that (i) $i_p >_{s_p} i_{p-1} >_{s_{p-1}} \dots >_{s_2} i_1 >_{s_1} i_0 >_{s_0} i_p$, and (ii) For each $q$, $0 \leq q \leq p$, $i_q>_{s_{q+1}} \emptyset$ and $i_q>_{s_{q}} \emptyset$, where $s_{p+1} = s_0$. A priority structure $>_S$ is \textbf{acyclic} if no cycles can be found.}

In our dynamic context, even requiring priorities to be acyclic, this is not possible when in period $t$ there are at least two teachers who were assigned to different schools in period $t-1$. To see this, consider a dynamic problem where schools' priorities are acyclic and agents $i$ and $j$ are assigned in period $t-1$ to $s$ and $s'$, respectively, in the related static problem, the priorities $\gg$ contain cycle $i \gg_s j \gg_{s'} i$.  Therefore, in our dynamic context, we cannot ensure that the $TRDA$ mechanism is strategy-proof, even in the more restrictive case where teachers' preferences are responsive.

Throughout this paper, we have assumed that teachers' preferences are time-invariant. However, it is reasonable to consider that these preferences may evolve due to previous experiences. Thus, a set of schools that in period $t$ was ranked as one of the preferred sets by a teacher may be ranked lower or even become unacceptable in a later period. The mechanism we have proposed here can be easily adapted to this situation. However, consider the case where there are only two teachers in the problem ($i$ and $j$), neither of whom holds a tenured position in any school, and three schools $s_1$, $s_2$, $s_3$. The second and third schools have one vacant position each, while the first has two vacant positions. All schools prioritize teacher $i$ over teacher $j$. Assume that teachers declare the following preferences:
\begin{center}
   \begin{tabular}{l}
  $\succ^t_i: s_1s_2, s_1s_3, s_1, s_2, s_3$\\
   $\succ^t_j: s_1s_3, s_1, s_3, s_2$
\end{tabular} 
\end{center}
The outcome of the $TRDA$ mechanism in period $t$ would then be:
$$\mu_{TR} = \begin{pmatrix} i & j \\ s_1s_2 & s_1s_3 \end{pmatrix}.$$
Now assume that in period $t+1$, teacher $i$ declares her (true) preference $\succ^{t+1}_i: s_1s_3, s_1s_2, s_1, s_2, s_3$.
Meanwhile, teacher $j$ maintains the same preference list over time. Under these conditions, the only dynamically rational matching is $\mu_{TR}.$ However, if agent $i$ misrepresent her preference in period $t$ by declaring $\succ^{t+1}_i$, she could manipulate the outcome because then she is assigned to schools $s_1$ and $s_3$ in both periods, resulting in a better outcome for teacher $i$ in period $t+1$. Moreover, the dynamic manipulation is obvious since $s_1s_3 \succ_i^{t+1} s_1s_2$. Therefore, the mechanism is obviously dynamically manipulable. Consequently, if teachers' preferences are allowed to change over time, the $TRDA$ mechanism can have obvious dynamic manipulations.

\appendix
\section{Appendix}\label{apendice de pruebas}

In this Appendix, we present the proofs of all our results. We first present Lemma \ref{estable en dinamico sii estable en estatico} that shows the relation between a dynamically stable matching for the dynamic problem, and a stable matching of a related static problem derived from priority $\gg^t_S.$ 

\begin{lemma}\label{estable en dinamico sii estable en estatico}
Let $\mathcal{P}^t$ be a dynamic problem, let $\mu^t\in \mathcal{M}^t$ be a matching, and let $\gg^t_S$ be the derived priorities. Then, 
 $\mu^t$ is dynamically stable at $\mathcal{P}^t$ if and only if it is stable at the related problem  $\mathcal{P}$.
\end{lemma}
\begin{proof}
$(\Longrightarrow)$ Let $\mu^t$ be a dynamically stable matching at $\mathcal{P}^t$. Thus, $\mu^t$ is non-wasteful and individually rational. So, to prove that $\mu^t$ is stable, we only need to show that $\mu^t$ eliminates all claims with respect to $\gg^t_S$. Assume otherwise, then there are $i\in I^t$ and $s\in S$ such that $i$ has a claim over $s$ with respect to $\gg_s^t$ at matching $\mu^t$. Thus, by definition, there is an other teacher $j\in \mu^t(s)$ such that $i\gg_{s}^t j$, $s\notin\mu^t (i)$, $s\in C_{i}(\mu^t(i)\cup \left\lbrace s \right \rbrace)$. We have two cases to consider:

\noindent \textbf{Case 1}: $\boldsymbol{j>_s^ti}$. Since $i\gg_s^tj$, we have that $i\in\mu^{t-1}(s)$ and $j\notin\mu^{t-1}(s)$, then $s\notin\mu^t(i)$ and $s\in\mu^{t-1}(i)$. Since $s\notin\mu^t(i)$, and the fact that $\mu^t$ is dynamically rational, then $s\notin C_i(\mu^t(i))$. Thus, by substitutability, $s\notin C_i(\mu^t(i)\cup\left\lbrace s\right\rbrace)$, which is a contradiction.

\noindent \textbf{Case 2}: $\boldsymbol{i>_s^tj}$. Since $i\gg_s^tj$, we have three sub-cases to consider:
        
             \textbf{Case 2.1: $\boldsymbol{i\in\mu^{t-1}(s)$ and $j\notin\mu^{t-1}(s)}$}. Then, there is a justified claim of $i$ over school $s$ with respect to priority $>^t_s$ at matching $\mu^t$.  
          
            \textbf{Case 2.2: $\boldsymbol{i,j\notin\mu^{t-1}(s)}$}. Then, there is a justified claim of $i$ over school $s$ with respect to priority $>^t_s$ at matching $\mu^t$.
            
             \textbf{Case 2.3: $\boldsymbol{i,j\in\mu ^{t-1}(s)}$}. Since $\mu^t$ is non-wasteful, $|\mu^t (s)|=q_s$ . Considering that $i\in\mu^{t-1}(s)$ and $i\notin \mu^t(s)$, there is $k\in\mu^t(s)\setminus\mu ^{t-1}(s)$ such that $i\gg^t_sk$. If $i>_s^tk$, by case 2.1, there is a justified claim of $i$ over school $s$ with respect to priority $>^t_s$ at matching $\mu^t$. If $k>_s^ti$, then by case 1, we have a contradiction. 

 Therefore, $\mu^t$ is stable at the related problem $\mathcal{P}$.
 
\noindent $(\Longleftarrow)$  Suppose that $\mu^t$ is stable to  $\mathcal{P}$ but is not dynamically stable. Then, we have two cases to consider:  

\noindent \textbf{Case 1: $\boldsymbol{\mu^t}$ is not dynamically rational}.  Note that $\mu^t$ is individually rational since it is stable. So, if $\mu^t$ is not dynamically rational, there is $i\in I_E^t$ such that $\mu^t(i)\nsucceq^B _i\mu ^{t-1}(i)$. Thus, $C_i(\mu^t(i)\cup \mu^{t-1}(i))\neq\mu^t(i)$. First, assume that there is $s\in C_i(\mu^t(i)\cup\mu^{t-1}(i))$ such that $s\notin\mu^{t}(i)$. Then, $s\in\mu^{t-1}(i).$ Hence, by substitutability,  $s\in C_i(\mu^t(i)\cup\{s \})$.  Since $\mu^t$ is non-wasteful ($|\mu^t (s)|=q_s$), we have that  $s\in\mu^{t-1}(i)$ and $s\notin \mu^t(i)$, then there is $j\in\mu^t(s)$ such that $j\notin \mu ^{t-1}(s)$. Hence, $i\gg _s^tj$, and this implies that $i$ has a claim over  $s$ with respect to priority $\gg^t_s$ at matching $\mu^t$. A contradiction since $\mu^t$ is, by hypothesis, stable at the related problem $\mathcal{P}.$ Second, assume that there is $s\in \mu^{t}(i)$ such that $s\notin C_i(\mu^t(i)\cup\mu^{t-1}(i))$. Then, there is $\overline{s}$ such that $\overline{s}\in \mu^{t-1}(i),$ $\overline{s}\notin \mu^t(i)$, and $\overline{s}\in C_i(\mu^t(i)\cup\{\overline{s}\}).$ Since $\mu^t$ is non-wasteful, there is $j\in \mu^t(\overline{s})$ such that $j\notin \mu^{t-1}(\overline{s}).$ Hence, $i\gg _{\overline{s}}^tj.$ This implies that $i$ has a claim over  $\overline{s}$ with respect to priority $\gg^t_{\overline{s}}$ at matching $\mu^t$. A contradiction since $\mu^t$ is, by hypothesis, stable at the related problem $\mathcal{P}.$

\noindent \textbf{Case 2: There is a justified claim with respect to priority $\boldsymbol{>^t_s$ at matching $\mu^t}$}. Assume that there are $i\in I^t$ and $s\in S$ such that $i$ have a justified claim over $s$ with respect to priority $>^t_s$ at matching $\mu^t$. Then, there is $j\in I^t$ such that $j\in\mu^t(s)$, $i>^t_s j$, $j\notin\mu^{t-1}(s)$, and $s\in C_i(\mu^t(i)\cup \{s\}$. This implies that $i\gg_s^tj$. Then, $i$ has a claim over  $s$ with respect to priority $\gg^t_s$ at matching $\mu^t$, contradicting that $\mu^t$ is, by hypothesis, stable at the related problem $\mathcal{P}.$ 

Therefore, $\mu^t$ is dynamically stable at problem $\mathcal{P}^t$.
\end{proof}

\noindent\begin{proof}[Proof of Theorem \ref{conjunto de estables dinamico no vacio}] 
 Let $\mathcal{P}^t$ be a dynamic problem. Consider the outcome of the teacher-proposing deferred acceptance mechanism applied to the related problem $\mathcal{P}$. By  \cite{GaleShapley1962}, such outcome is a stable matching for $\mathcal{P}$ and by Lemma \ref{estable en dinamico sii estable en estatico} it is a dynamically stable matching for $\mathcal{P}^t$. 
\end{proof}


\noindent\begin{proof}[Proof of Theorem \ref{TRDA minumica reclamos injustificados}]
Given a dynamic problem $\mathcal{P}^t$, by Theorem \ref{conjunto de estables dinamico no vacio}, we know that the matching $\mu_{TR}$ is dynamically stable. 
If $\mathcal{U}(\mu_{TR})=\emptyset$ the proof is complete.
Otherwise, if $\mu_{TR}$ does not minimize unjustified claims, then there is a dynamically stable matching $\mu^t$ such that $\mathcal{U}(\mu^t)\subsetneq \mathcal{U}(\mu_{TR})$. 

Since matching $\mu_{TR}$ is constrained-efficient (for teachers), $\mu_{TR}(i)\succeq^B_i \mu^t(i)$  for each $i\in I^t$ and there is a teacher $h$ such that $\mu_{TR}(h)\succ^B_h \mu^t(h)$.

Suppose that there is $i\in I^t$ and $s\in S$, such that $(i,s)\in\mathcal{U}(\mu_{TR})$ and $(i,s)\notin \mathcal{U}(\mu^t)$. Then, $s\notin \mu_{TR}(i)$, $s\in C_i(\mu_{TR}(i)\cup \{s\})$ and there is a teacher $j\in\mu_{TR}(s)$ such that $j\in\mu^{t-1}(i)$ and $i>^t_sj$. Note that $(i,s)\notin \mathcal{U} (\mu^t)$ can happen for three reasons: 

\noindent \textbf{Case 1: there is $\boldsymbol{j\in\mu^t(s)$ such that $i>^t_sj $ and $j\in\mu^{t-1}(s)}$}. Then, if teacher $i$ has a justified claim over $s$ at $\mu^t$, then $\mu^t$ is not a dynamically stable matching, which contradicts how we had taken $\mu^t$.

\noindent \textbf{Case 2: $\boldsymbol{s\notin C_i(\mu^t(i)\cup\{s\})}$}. By substitutability we know that $s\notin C_i(\mu_{TR}(i) \cup \mu^t \cup \{s\})$. Note that by the constrained-efficiency of $\mu_{TR}$ and the path-independence property of $C_i$ we have that:
\begin{center}
     $C_i(\mu_{TR}(i) \cup \mu^t(i) \cup \{s\})=C_i(C_i(\mu_{TR}(i) \cup \mu^t(i))\cup \{s\})=C_i(\mu_{TR}(i) \cup\{s\})$.
\end{center}
Then, $s\notin C_i(\mu_{TR}(i) \cup \{s\})$ which is a contradiction.

\noindent \textbf{Case 3: $\boldsymbol{s\in \mu^t(i)}$}. As $(i,s)\in \mathcal{U}(\mu_{TR})$ we have that $s\in C_i(\mu_{TR}(i)\cup\{s\})$. Given that $\mu_{TR}$ is dynamically rational and  constrained-efficient we have that
$$
    C_i(\mu_{TR}(i)\cup\{s\})=C_i(C_i(\mu_{TR}(i)\cup \mu^t(i))\cup\{s\}).$$
By the path-independence property of $C_i$, 
$$C_i(C_i(\mu_{TR}(i)\cup \mu^t(i))\cup\{s\})=C_i(\mu_{TR}(i)\cup \mu^t(i) \cup\{s\}).$$
Since $s\in \mu^t(i)$ and the constrained-efficiency of $\mu_{TR},$ we have
$$C_i(\mu_{TR}(i)\cup \mu^t(i) \cup\{s\})=C_i(\mu_{TR}(i)\cup \mu^t(i))=C_i(\mu_{TR}(i)),$$
 implying that $C_i(\mu_{TR}(i)\cup\{s\})=C_i(\mu_{TR}(i)).$
 Then $s\in C_i(\mu_{TR}(i))$ which contradicts previous statements. 

 So, we proved that $\mu_{TR}$ minimizes unjustified claims.
\end{proof}


\noindent \begin{proof}[Proof of Lemma \ref{lema s-truncacion de C_i}]
    Let $(i,s)$ be a teacher-school pair, and let $S',S''\subseteq S$. Let $C_i$ be a path-independent choice function. By the Definition \ref{defino s-truncacion} and the path-independence of $C_i$,   
    \begin{equation}\label{ecu1 lema s-truc}
    \begin{split}
        &C_i^s(S'\cup S'')=C_i\left((S'\cup S'')\setminus \{s\}\right)=C_i\left((S'\setminus \{s\})\cup (S''\setminus \{s\})\right)=\\
      &=C_i\left( C_i(S'\setminus \{s\})\cup (S''\setminus \{s\})\right)=C_i\left( C^s_i(S')\cup (S''\setminus \{s\})\right).
    \end{split}
          \end{equation}
   By definition of choice function $C^s_i$ and $C_i$, $C^s_i(S')=C_i(S'\setminus\{s\})\subseteq S'\setminus\{s\}.$
   This last equation together with Definition \ref{defino s-truncacion}, implies that 
   \begin{equation} \label{ecu2 lema s-truc}
   \begin{split}
       &C_i\left( C^s_i(S')\cup (S''\setminus \{s\})\right)=C_i\left( (C^s_i(S')\setminus \{s\} )\cup (S''\setminus \{s\})\right)=\\
       &=C_i\left( (C^s_i(S') \cup S'')\setminus \{s\}\right)=C^s_i\left( C^s_i(S') \cup S''\right).
   \end{split}       
   \end{equation}
   By \eqref{ecu1 lema s-truc} and \eqref{ecu2 lema s-truc} we have that $C_i^s(S'\cup S'')=C^s_i\left( C^s_i(S') \cup S''\right)$ and, therefore,  choice function $C^s_i$ is path-independent.
 \end{proof}

The following two lemmata are the key ingredients for proving that the $TREADA$ mechanism is (non-constrained) efficient if all teachers consent. The proofs of these lemmata share the same spirit as those in \cite{kesten2010school}, but have been adapted to our context using techniques appropriate for path-independent choice functions.
\begin{lemma}\label{lema 2 del eada}
For the domain of dynamic problems with path-independent choice functions, the matching produced at the end of round~$r$ ($r \geq 1$) by the $TREADA$ mechanism assigns each teacher to a set of schools that is at least as good for her, according to Blair’s partial order, as the set she was assigned at the end of round~$r-1$.
\end{lemma}
\begin{proof}
Assume there is a problem, a round $r$ ($r \geq 1$) of the $TREADA$ mechanism, and a teacher $i_1$ such that the set of schools assigned to $i_1$ in round $r$ is not Blair-preferred to the set assigned in round $r-1$. That is, $\mu_{r}(i_1) \not\succ^B_{i_1} \mu_{r-1}(i_1)$, where $\mu_{r}(i_1)$ denotes the set of schools assigned to teacher $i_1$ by the matching at the end of round $r$, and $\mu_{r-1}(i_1)$ denotes the set assigned at the end of round $r-1$.

\noindent\textbf{Claim:} \textbf{$\boldsymbol{\mu_{r}(i_1) \not\succ^B_{i_1} \mu_{r-1}(i_1)}$ implies that there is a school $\boldsymbol{s_1 \in \mu_{r-1}(i_1)\setminus \mu_{r}(i_1)}$ and a teacher $\boldsymbol{i_2\in \mu_{r}(s_1)\setminus \mu_{r-1}(s_1)}$}. 
Assume not, then $\mu_{r-1}(i_1)\subsetneq \mu_{r}(i_1).$ Since $\mu_{r}(i_1) \not\succ^B_{i_1} \mu_{r-1}(i_1)$, we have that  $\mu_{r}(i_1)\neq C_{i_1}(\mu_{r}(i_1) \cup \mu_{r-1}(i_1))=C_{i_1}(\mu_{r}(i_1))$, contradicting the individual rationality of $\mu_{r}$. Thus, there is a school $s_1 \in \mu_{r-1}(i_1)$ such that $s_1 \notin \mu_r(i_1)$, that is, $i_1$ was rejected by school $s_1$ at round $r$ and another teacher $i_2 \in I \setminus \{i_1\}$ who had not previously applied to school $s_1$ is temporarily accepted at round $r$ and proving the claim. 

Since $i_2$ is temporarily accepted in round $r$ after being rejected by a school to which she was assigned under $\mu_{r-1}$, it follows that $\mu_{r-1}(i_2) \succ^B_{i_2} \mu_r(i_2)$. By similar reasoning to the claim, there is a school $s_2 \in \mu_{r-1}(i_2)$ such that $s_2 \notin \mu_r(i_2)$ and there is a teacher $i_3 \in I \setminus \{i_1, i_2\}$ such that $s_2 \notin \mu_{r-1}(i_3)$ and $s_2 \in \mu_r(i_3)$. Continuing recursively in this way, there is a school $s_{k-1} \in \mu_{r-1}(i_{k-1})$ such that $s_{k-1} \notin \mu_r(i_{k-1})$ and there is a teacher $i_k \in I \setminus \{i_1, i_2, \dots, i_{k-1}\}$ such that $s_{k-1} \notin \mu_{r-1}(i_k)$ and $s_{k-1} \in \mu_r(i_k)$. Therefore, by finiteness of the set $I$, there is a teacher $i_k \in I \setminus \{i_1, \dots, i_{k-1}\}$ who is the first teacher to apply to a set of schools less Blair-preferred than the set obtained in round $r-1$, i.e., $\mu_{r-1}(i_k) \succ^B_{i_k} \mu_r(i_k)$. Now, we analyze two possibilities:

\noindent\textbf{Case 1: $\boldsymbol{i_k$ is not an interrupter in round $r-1}$}. Then, the choice function of teacher $i_k$ remain the same in both rounds $r$ and $r-1$, and by similar reasoning to the claim, there are a school $s \in \mu_{r-1}(i_k) \setminus \mu_r(i_k)$ and a teacher $i \in I \setminus \{i_1, \dots, i_k\}$ such that $s \notin \mu_{r-1}(i)$ and $s \in \mu_r(i)$. This contradicts that $i_k$ is the first teacher to apply to a less Blair-preferred set of schools than the one obtained in round $r-1$.

\noindent\textbf{Case 2: $\boldsymbol{i_k$ is an interrupter in round $r-1}$}. Then, teacher $i_k$, instead of applying to a set of schools that includes a school for which she is an interrupter, applies to the next best set that does not include the schools for which $i_k$ is an interrupter according to Blair's ordering, say $S^\star$. Then, $S^\star \succeq^B_{i_k} \mu_{r-1}(i_k) \succ^B_{i_k} \mu_r(i_k).$
Therefore, by similar reasoning to the claim, there is a school $s_k \in \mu_{r-1}(i_k) \setminus \mu_r(i_k)$ and a teacher $i \in I \setminus \{i_1, \dots, i_k\}$ such that $i \in \mu_r(s_k)$ and $i \notin \mu_{r-1}(s_k).$ This contradicts that $i_k$ is the first teacher to apply to a less Blair-preferred set of schools than the one obtained in round $r-1$.

Thus, by Cases 1 and 2, we have that the matching obtained at the end of round $r$ of the $TREADA$ mechanism, with $r \geq 1$, places each teacher in a set of schools that is at least as good for them as the set of schools they were placed in at the end of round $r-1$, according to Blair's ordering.
\end{proof}

\begin{lemma}\label{lema 3 del eada}
For the domain of dynamic problems with path-independent choice functions,  the outcome of the $TREADA$ mechanism is efficient when all teachers consent.
\end{lemma}
\begin{proof}
Assume, by contradiction, that there is a problem for which the matching selected by the $TREADA$ mechanism is not efficient. Call this matching $\mu$. Then, there is another matching $\nu$ that Blair-dominates $\mu$. Assume the algorithm ends in $R \geq 1$ rounds. For $r \in \{1, 2, \dots, R\}$, let $\mu_r$ be the matching obtained at the end of round $r$ of the $TREADA$ mechanism. By Lemma \ref{lema 2 del eada}, $\nu(i) \succeq^B_i \mu(i) \succeq^B_i \mu_r(i)$ for all $i \in I$, and there is a teacher $j \in I$ such that $\nu(j) \succ^B_j \mu(j) \succeq^B_j \mu_r(j)$.  
First, we will see that for  $\nu$, there is no interrupter at any round.  Assume that under $\nu$, there is an interrupter $i_1$ from round 1 who is assigned to $s_1$, for which she is an interrupter. That is, $i_1 \in \nu(s_1)\setminus \mu_1(s_1)$. Furthermore, $ |\mu_1(s_1)| = q_{s_1}$; otherwise, teacher $i_1$ should not have been rejected by $s_1$ at $\mu_1$. Then, there is a teacher $i_2$ such that $i_2 \in \mu_1(s_1)\setminus \nu(s_1)$, and $\nu(i_2) \succ^B_{i_2} \mu_1(i_2).$  Since $s_1 \in \mu_1(i_2) \setminus \nu(i_2)$, there is a school $s_2 \in \nu(i_2) \setminus \mu_1(i_2)$, such that   $ |\mu_1(s_2)| = q_{s_2}$. Otherwise, teacher $i_2$ should not have been rejected by $s_2$ at $\mu_1$.

In the same manner, there is a teacher $i_3$ such that $i_3 \in \mu_1(s_2)\setminus \nu(s_2)$. By an analogous argument as before, since $s_2 \in \mu_1(i_3) \setminus \nu(i_3)$, there is a school $s_3 \in \nu(i_3) \setminus \mu_1(i_3)$ such that $\nu(i_3) \succ^B_{i_3} \mu_1(i_3).$ Continuing in this recursive reasoning, there is a teacher $i_k$ such that $i_k \in \mu_1(s_{k-1})\setminus \nu(s_{k-1})$. By a similar argument as before, since $s_{k-1} \in \mu_1(i_k) \setminus \nu(i_k)$, there is a school $s_k \in \nu(i_k) \setminus \mu_1(i_k)$ such that $\nu(i_k) \succ^B_{i_k} \mu_1(i_k).$ Furthermore, following the same reasoning as before, $ |\mu_1(s_k)| = q_{s_k}.$ Note that by the finiteness of the set of schools, we eventually reach school $s_1$. Assume, w.l.o.g.,  $s_k = s_1.$

Thus, forming a cycle of teachers $(i_1, i_2, \dots, i_k)$, $k \geq 2$, and a cycle of schools such that each teacher wants to change her set of schools by incorporating a new school where the next teacher in the cycle is assigned under matching $\mu_1$. Since $ |\mu_1(s_\ell)| = q_{s_\ell}$ for each $\ell=1,\ldots,k$, each time a teacher incorporates a new school, she must let go of another school in her current set.

Considering the $TRDA$ mechanism from the round 0 of $TREADA$ mechanism, let $i_p \in \{i_1, i_2, \dots, i_k\}$ be the teacher in the cycle (or one of them, if there are more than one) who is the last to apply to the school $s_p$ with which she ends up at the end of this round, i.e., $i_p \in \nu(s_p)\setminus \mu_1(s_p)$. Teacher $i_{p-1}$, who wants to incorporate school $s_p$ (and let go of some other school) was already rejected at some previous step, implying that when $i_p$ applies to $s_p$, $|\mu_1(s_p)| = q_{s_p}$. Therefore, some teacher $i'$ is rejected. Consequently, $i'$ is an interrupter and, thus, is rejected at a later step after $i_1$ was rejected by $s_1$ for which $i_1$ is an interrupter. But then, teacher $i_1$ cannot be an interrupter in round 1.

Assume that under matching $\nu$, no interrupter from round $k$, $1 \leq k \leq r-1$, is assigned to a school for which she is an interrupter in round $k$. We want to show that under matching $\nu$, no interrupter from round $r$ is assigned to a school for which she is an interrupter in round $r$. Since $\nu$ Blair-dominates each $\mu_r$, by using the same argument as before, there is a cycle of teachers $(i'_1, \dots, i'_{k'})$ such that each teacher wants to incorporate a school to which the next teacher in the cycle is assigned under matching $\mu_r$ (and let go of some other school). By the induction hypothesis, none of the teachers in this cycle are interrupters for the schools they want to incorporate. Thus, for each teacher in the cycle, there is a step of the $TRDA$ mechanism executed in round $r-1$ where the teacher is rejected by the school she prefers. But then we can apply the same argument we used earlier to conclude that teacher $i'_1$ cannot be an interrupter in round $r$.

At the end of round $R$, there are no remaining interrupters, and we obtain matching $\mu$. Since matching $\nu$ Blair-dominates matching $\mu$, again, there is a cycle of teachers $(i''_1, i''_2, \dots, i''_{k''})$, $k \geq 2$, who want to incorporate a school (and let go of some other school) where the next teacher in the cycle is assigned (for teacher $i''_k$, it is $i''_1$) under matching $\mu$. Note that no teacher in $(i''_1, i''_2, \dots, i''_{k''})$  in any round can be an interrupter for the school to which the next teacher in the cycle is assigned under matching $\mu$. Hence, each teacher in $(i''_1, i''_2, \dots, i''_{k''})$ applies to the school to which the next teacher in the cycle is assigned in some step of the $TRDA$ mechanism in round $R$. Let $i'' \in (i''_1, i''_2, \dots, i''_{k''})$ be the last teacher in the cycle to be rejected by a school to which the next teacher is assigned. When teacher $i''$ applies to the school where she ends up at the end of round $R$,  some teacher $i'''$ was rejected by that school in an earlier step. This occurs because any teacher in $(i''_1, i''_2, \dots, i''_{k''})$, who incorporates this school (and let go of some other school)  was previously rejected by the school at an earlier stage. Therefore, teacher $i'''$ is an interrupter, which contradicts $R$ being the last round. 
 \end{proof}

\noindent \begin{proof}[Proof of Theorem \ref{teorema TREADA eficiente}]
    That the $TREADA$ Blair-dominates the $TRDA$ mechanism is shown in Lemma \ref{lema 2 del eada}. It is clear that if there is no consenting teacher, the mechanism ends in Round $0$, thus the $TREADA$ mechanism and the $TRDA$ mechanism yield the same outcome.  When all teachers consent, the $TREADA$ mechanism's efficiency is proved in Lemma \ref{lema 3 del eada}.
\end{proof}

Now, we present proof that the $TREADA$ and all stable-dominating mechanisms are dynamically not obviously manipulable.\medskip

\noindent\begin{proof}[Proof of Theorem \ref{teorema TREADA es NOM}]
  Let $\mathcal{E}=(\{I^t\}_t, S, \{\succ_{I^t}\}_t, \{>_S^t\}_t,q_S,\mu^0,\varphi)$ be a economy where $\varphi$ is the $TRDA$ mechanism. Denote by $I_{t'}^t$ the set of teachers at period $t$ who joined the economy at period $t'\leq t$. If the priority structure is lexicographic by tenure, then the outcome of the $TRDA$ mechanism is the same outcome of the following procedure:
    
    \textbf{Step $\boldsymbol{1}$}: the teachers in $I_1^t$ propose to her best subset of schools under the teacher-proposing deferred acceptance mechanism. The capacities of the schools are updated by subtracting the positions that are filled.
    
    \textbf{Step $\boldsymbol{t'\leq t}$}: the teachers in $I_{t'}^t$ propose to their best subset of schools (with their unfilled capacities) under the teacher-proposing deferred acceptance mechanism and a matching is produced. The capacities of the schools are updated by subtracting the positions that are filled.

Next, we show that both mechanisms have the same outcome under the assumption that schools' priorities are lexicographic by tenure. 
  Assume that this is not the case. By the tenured nature of our model, the outcome of the $TRDA$ mechanism and the outcome of the previously described mechanisms will be different when there are a period $t$, two teachers $i,j\in I^t $, and a school $s\in S$ such that $i >^t_s j$, $j \gg^t_s i$ and $i$ proposes to $s$.
Thus, there are two cases to consider. Teachers $i,j$ join the economy at the same period or teacher $i$ joins the economy before teacher $j$.\footnote{That \emph{teacher $i$ joins the economy at period $t$} means that there is a school $s$ such that $i\in \mu^t(s)$ and $i\notin I^t_E.$} In both cases, we have that $j\in\mu^{t-1}(s)$ and $i\notin \mu^{t-1}(s).$ This implies that $i,j\in I^t_E.$ 

The fact that $i$ proposes to $s$ implies that $s\notin \mu^{t}(i)$ and $s\in C_i(\mu^{t}(i)\cup \{s\})$.
Since $i\in I^{t}_E$ and $\mu^{t}$ is dynamically rational, $\mu^{t}(i)\succeq_i^B \mu^{t-1}(i).$ Then, 
$s\in C_i(C_i(\mu^{t}(i)\cup \mu^{t-1}(i) ) \cup \{s\})$. By the path-independence property of $C_i$, we have that $s\in C_i(\mu^{t}(i)\cup\mu^{t-1}(i)\cup \{s\}).$ By substitutability, 
\begin{equation}\label{ecu 0 SP}
    s\in C_i(\mu^{t-1}(i)\cup \{s\}).
\end{equation}
Let $k$ be the period when teacher $j$ was assigned to $s$ for the first time and $s\notin \mu^k(i)$. Then,  $j\notin \mu^{k-1}(s).$ Moreover, since $i,j$ join the economy at the same period or teacher $i$ joins the economy before teacher $j$, we have that $i\in I^{k+1}_E.$
Given that $\mu^{\ell}$ is dynamically rational for each $\ell=k+1,\ldots,t$, we have that
\begin{equation}\label{ecu 1 SP}
    \mu^{t}(i)\succeq_i^B \mu^{t-1}(i)\succeq_i^B\ldots \succeq_i^B\mu^{k+1}(i)\succeq_i^B\mu^{k}(i).
\end{equation}
Hence,  by \eqref{ecu 0 SP} and \eqref{ecu 1 SP}, we have $s\in C_i(C_i(\mu^{t-1}(i)\cup\mu^k(i))\cup\{s\} )$. By the path-independence property of $C_i$, we have that $s\in C_i(\mu^{t-1}(i)\cup\mu^k(i)\cup\{s\} )$. Then by substitutability, $s\in C_i(\mu^k(i)\cup\{s\})$. Thus, there is a justified claim of $i$ over $s$ at $\mu^k$, since $s\notin \mu^k(i)$ and $i >^t_s j$ (which is the same that $i >^k_s j$). So, $\mu^k$ is not a dynamically stable matching, which is a contradiction. 
 Therefore, both mechanisms generate the same outcome.

Given that priorities are assumed to be lexicographic by tenure, a teacher in $I^t_{t'}$ only competes with other teachers within the same set. Teachers in this set, under the previous mechanism, cannot alter the unfilled positions of schools after step $t' - 1.$ Thus, in each period, considering the teachers who enter the problem during that period and the remaining unfilled positions of the schools, we can treat this situation as a new static problem. Consequently, by Theorem \ref{teorema 2 arribi} (i), this mechanism does not exhibit obvious manipulations. Furthermore, by Theorem \ref{teorema 2 arribi} (ii), the $TREADA$ mechanism, together with all dynamically stable-dominating mechanisms, are not obviously dynamically manipulable.
\end{proof}


\begin{thebibliography}{40}
\newcommand{\enquote}[1]{``#1''}
\expandafter\ifx\csname natexlab\endcsname\relax\def\natexlab#1{#1}\fi

\bibitem[\protect\citeauthoryear{Abdulkadiro{\u{g}}lu, Che, and
  Yasuda}{Abdulkadiro{\u{g}}lu et~al.}{2011}]{Abdulkadiroglu2011}
\textsc{Abdulkadiro{\u{g}}lu, A., Y.-K. Che, and Y.~Yasuda} (2011):
  \enquote{Resolving conflicting preferences in school choice: The "Boston
  Mechanism" reconsidered,} \emph{American Economic Review}, 101, 399--410.

\bibitem[\protect\citeauthoryear{Abdulkadiro{\u{g}}lu, Pathak, and
  Roth}{Abdulkadiro{\u{g}}lu et~al.}{2004}]{Abdulkadiroglu2004}
\textsc{Abdulkadiro{\u{g}}lu, A., P.~Pathak, and A.~Roth} (2004): \enquote{The
  New York City High School Match,} \emph{American Economic Review}, 95,
  364--367.

\bibitem[\protect\citeauthoryear{Abdulkadiro{\u{g}}lu, Pathak, and
  Roth}{Abdulkadiro{\u{g}}lu et~al.}{2009}]{Abdulkadiroglu2009}
---\hspace{-.1pt}---\hspace{-.1pt}--- (2009): \enquote{Strategy-proofness
  versus Efficiency in Matching with Indifferences: Redesigning the NYC High
  School Match,} \emph{American Economic Review}, 99, 1954--1978.

\bibitem[\protect\citeauthoryear{Abdulkadiro{\u{g}}lu, Pathak, Roth, and
  S{\"{o}}nmez}{Abdulkadiro{\u{g}}lu et~al.}{2005}]{Abdulkadiroglu2005}
\textsc{Abdulkadiro{\u{g}}lu, A., P.~Pathak, A.~Roth, and T.~S{\"{o}}nmez}
  (2005): \enquote{The Boston public school match,} \emph{American Economic
  Review}, 95, 368--371.

\bibitem[\protect\citeauthoryear{Abdulkadiro{\u{g}}lu and
  S{\"o}nmez}{Abdulkadiro{\u{g}}lu and
  S{\"o}nmez}{2003}]{abdulkadirouglu2003school}
\textsc{Abdulkadiro{\u{g}}lu, A. and T.~S{\"o}nmez} (2003): \enquote{School
  choice: A mechanism design approach,} \emph{American economic review}, 93,
  729--747.

\bibitem[\protect\citeauthoryear{Aizerman and Malishevski}{Aizerman and
  Malishevski}{1981}]{aizerman1981general}
\textsc{Aizerman, M. and A.~Malishevski} (1981): \enquote{General theory of
  best variants choice: Some aspects,} \emph{IEEE Transactions on Automatic
  Control}, 26, 1030--1040.

\bibitem[\protect\citeauthoryear{Alva and Manjunath}{Alva and
  Manjunath}{2019}]{alva2019strategy}
\textsc{Alva, S. and V.~Manjunath} (2019): \enquote{Strategy-proof
  Pareto-improvement,} \emph{Journal of Economic Theory}, 181, 121--142.

\bibitem[\protect\citeauthoryear{Arribillaga and Risma}{Arribillaga and
  Risma}{2025{\natexlab{a}}}]{arribillaga2024many-yo-many-obvious}
\textsc{Arribillaga, R.~P. and E.~P. Risma} (2025{\natexlab{a}}):
  \enquote{Obvious manipulations in many-to-many matching with and without
  contracts,} \emph{Working paper}.

\bibitem[\protect\citeauthoryear{Arribillaga and Risma}{Arribillaga and
  Risma}{2025{\natexlab{b}}}]{arribillaga2023obvious}
---\hspace{-.1pt}---\hspace{-.1pt}--- (2025{\natexlab{b}}): \enquote{Obvious
  manipulations in matching with and without contracts,} \emph{Games and
  Economic Behavior}, 151, 70--81.

\bibitem[\protect\citeauthoryear{Blair}{Blair}{1988}]{Blair1988}
\textsc{Blair, C.} (1988): \enquote{The Lattice Structure of the Set of Stable
  Matchings with Multiple Partners,} \emph{Mathematics of Operations Research},
  13, 619--628.

\bibitem[\protect\citeauthoryear{Bonifacio, Juarez, Neme, and Oviedo}{Bonifacio
  et~al.}{2022}]{bonifacio2022cycles}
\textsc{Bonifacio, A.~G., N.~Juarez, P.~Neme, and J.~Oviedo} (2022):
  \enquote{Cycles to compute the full set of many-to-many stable matchings,}
  \emph{Mathematical Social Sciences}, 117, 20--29.

\bibitem[\protect\citeauthoryear{Chambers and Yenmez}{Chambers and
  Yenmez}{2017}]{chamb2017choice}
\textsc{Chambers, C.~P. and M.~B. Yenmez} (2017): \enquote{Choice and
  matching,} \emph{American Economic Journal: Microeconomics}, 9, 126--147.

\bibitem[\protect\citeauthoryear{Chen and Pereyra}{Chen and
  Pereyra}{2019}]{chen2019self}
\textsc{Chen, L. and J.~S. Pereyra} (2019): \enquote{Self-selection in school
  choice,} \emph{Games and Economic Behavior}, 117, 59--81.

\bibitem[\protect\citeauthoryear{Chen and M{\"o}ller}{Chen and
  M{\"o}ller}{2024}]{chen2024regret}
\textsc{Chen, Y. and M.~M{\"o}ller} (2024): \enquote{Regret-free truth-telling
  in school choice with consent,} \emph{Theoretical Economics}, 19, 635--666.

\bibitem[\protect\citeauthoryear{Combe, Tercieux, and Terrier}{Combe
  et~al.}{2022}]{combe2022design}
\textsc{Combe, J., O.~Tercieux, and C.~Terrier} (2022): \enquote{The design of
  teacher assignment: Theory and evidence,} \emph{The Review of Economic
  Studies}, 89, 3154--3222.

\bibitem[\protect\citeauthoryear{Doval}{Doval}{2022}]{doval2022dynamicallstable}
\textsc{Doval, L.} (2022): \enquote{Dynamically stable matching,}
  \emph{Theoretical Economics}, 17, 687--724.

\bibitem[\protect\citeauthoryear{Echenique and Oviedo}{Echenique and
  Oviedo}{2004}]{Echenique2004}
\textsc{Echenique, F. and J.~Oviedo} (2004): \enquote{Core many-to-one
  matchings by fixed-point methods,} \emph{Journal of Economic Theory}, 115,
  358--376.

\bibitem[\protect\citeauthoryear{Echenique and Oviedo}{Echenique and
  Oviedo}{2006}]{Echenique2006}
---\hspace{-.1pt}---\hspace{-.1pt}--- (2006): \enquote{A theory of stability in
  many-to-many matching markets,} \emph{Theoretical Economics}, 1, 233--273.

\bibitem[\protect\citeauthoryear{Fleiner}{Fleiner}{2003}]{Fleiner2003}
\textsc{Fleiner, T.} (2003): \enquote{A Fixed-Point Approach to Stable
  Matchings and Some Applications,} \emph{Mathematics of Operations Research},
  28, 103--126.

\bibitem[\protect\citeauthoryear{Gale and Shapley}{Gale and
  Shapley}{1962}]{GaleShapley1962}
\textsc{Gale, D. and L.~Shapley} (1962): \enquote{College admissions and the
  stability of marriage,} \emph{American Mathematical Monthly}, 69, 9--15.

\bibitem[\protect\citeauthoryear{Hatfield and Milgrom}{Hatfield and
  Milgrom}{2005}]{hatfield2005matching}
\textsc{Hatfield, J. and P.~Milgrom} (2005): \enquote{Matching with contracts,}
  \emph{American Economic Review}, 95, 913--935.

\bibitem[\protect\citeauthoryear{Hirata and Kasuya}{Hirata and
  Kasuya}{2017}]{hirata2017stable}
\textsc{Hirata, D. and Y.~Kasuya} (2017): \enquote{On stable and strategy-proof
  rules in matching markets with contracts,} \emph{Journal of Economic Theory},
  168, 27--43.

\bibitem[\protect\citeauthoryear{Iwase}{Iwase}{2022}]{iwase2022equivalence}
\textsc{Iwase, Y.} (2022): \enquote{Equivalence theorem in matching with
  contracts,} \emph{Review of Economic Design}, 1--9.

\bibitem[\protect\citeauthoryear{Kesten}{Kesten}{2010}]{kesten2010school}
\textsc{Kesten, O.} (2010): \enquote{School choice with consent,} \emph{The
  Quarterly Journal of Economics}, 125, 1297--1348.

\bibitem[\protect\citeauthoryear{Klijn, Pais, and Vorsatz}{Klijn
  et~al.}{2013}]{Klijn2013}
\textsc{Klijn, F., J.~Pais, and M.~Vorsatz} (2013): \enquote{Preference
  intensities and risk aversion in school choice: A laboratory experiment,}
  \emph{Experimental Economics}, 16, 1--22.

\bibitem[\protect\citeauthoryear{Kojima}{Kojima}{2011}]{Kojima2011}
\textsc{Kojima, F.} (2011): \enquote{Robust stability in matching markets,}
  \emph{Theoretical Economics}, 6, 257--267.

\bibitem[\protect\citeauthoryear{Mart{\'\i}nez, Mass{\'o}, Neme, and
  Oviedo}{Mart{\'\i}nez et~al.}{2004{\natexlab{a}}}]{Martinez2004}
\textsc{Mart{\'\i}nez, R., J.~Mass{\'o}, A.~Neme, and J.~Oviedo}
  (2004{\natexlab{a}}): \enquote{An algorithm to compute the set of
  many-to-many stable matchings,} \emph{Mathematical Social Sciences}, 47,
  187--210.

\bibitem[\protect\citeauthoryear{Mart{\'\i}nez, Mass{\'o}, Neme, and
  Oviedo}{Mart{\'\i}nez et~al.}{2004{\natexlab{b}}}]{martinez2004group}
---\hspace{-.1pt}---\hspace{-.1pt}--- (2004{\natexlab{b}}): \enquote{On group
  strategy-proof mechanisms for a many-to-one matching model,}
  \emph{International Journal of Game Theory}, 33, 115--128.

\bibitem[\protect\citeauthoryear{Mart{\'\i}nez, Mass{\'o}, Neme, and
  Oviedo}{Mart{\'\i}nez et~al.}{2008}]{martinez2008invariance}
---\hspace{-.1pt}---\hspace{-.1pt}--- (2008): \enquote{On the invariance of the
  set of stable matchings with respect to substitutable preference profiles,}
  \emph{International Journal of Game Theory}, 36, 497--518.

\bibitem[\protect\citeauthoryear{Nicol{\`o}, Salmaso, and Saulle}{Nicol{\`o}
  et~al.}{2023}]{nicolo2023dynamic}
\textsc{Nicol{\`o}, A., P.~Salmaso, and R.~Saulle} (2023): \enquote{Dynamic
  one-sided matching,} \emph{Available at SSRN 4352130}.

\bibitem[\protect\citeauthoryear{Pereyra}{Pereyra}{2013}]{pereyra2013dynamic}
\textsc{Pereyra, J.~S.} (2013): \enquote{A dynamic school choice model,}
  \emph{Games and Economic Behavior}, 80, 100--114.

\bibitem[\protect\citeauthoryear{Romero-Medina and Triossi}{Romero-Medina and
  Triossi}{2021}]{romero2021two}
\textsc{Romero-Medina, A. and M.~Triossi} (2021): \enquote{Two-sided
  strategy-proofness in many-to-many matching markets,} \emph{International
  Journal of Game Theory}, 50, 105--118.

\bibitem[\protect\citeauthoryear{Roth}{Roth}{1984}]{Roth1984b}
\textsc{Roth, A.} (1984): \enquote{Stability and Polarization of Interests in
  Job Matching,} \emph{Econometrica}, 52, 47 -- 57.

\bibitem[\protect\citeauthoryear{Roth}{Roth}{1985{\natexlab{a}}}]{Roth1985a}
---\hspace{-.1pt}---\hspace{-.1pt}--- (1985{\natexlab{a}}): \enquote{The
  college admissions problem is not equivalent to the marriage problem,}
  \emph{Journal of Economic Theory}, 36, 277--288.

\bibitem[\protect\citeauthoryear{Roth}{Roth}{1985{\natexlab{b}}}]{Roth1985}
---\hspace{-.1pt}---\hspace{-.1pt}--- (1985{\natexlab{b}}): \enquote{Conflict
  and Coincidence of Interest in Job Matching: Some New Results and Open
  Questions,} \emph{Mathematics of Operations Research}, 10, 379--389.

\bibitem[\protect\citeauthoryear{Roth and Sotomayor}{Roth and
  Sotomayor}{1990}]{RothSotomayor90}
\textsc{Roth, A. and M.~Sotomayor} (1990): \emph{Two-sided matching}, Cambridge
  University Press.

\bibitem[\protect\citeauthoryear{Sakai}{Sakai}{2011}]{sakai2011note}
\textsc{Sakai, T.} (2011): \enquote{A note on strategy-proofness from the
  doctor side in matching with contracts,} \emph{Review of Economic Design},
  15, 337--342.

\bibitem[\protect\citeauthoryear{Saulle, Nicol{\`o}, and Salmaso}{Saulle
  et~al.}{2024}]{saulle2024rationalizable}
\textsc{Saulle, R., A.~Nicol{\`o}, and P.~Salmaso} (2024):
  \enquote{Rationalizable Conjectures in Dynamic Matching,} .

\bibitem[\protect\citeauthoryear{Tang and Zhang}{Tang and
  Zhang}{2021}]{tang2021weak}
\textsc{Tang, Q. and Y.~Zhang} (2021): \enquote{Weak stability and Pareto
  efficiency in school choice,} \emph{Economic Theory}, 71, 533--552.

\bibitem[\protect\citeauthoryear{Troyan and Morrill}{Troyan and
  Morrill}{2020}]{troyan2020obvious}
\textsc{Troyan, P. and T.~Morrill} (2020): \enquote{Obvious manipulations,}
  \emph{Journal of Economic Theory}, 185, 104970.

\end{thebibliography}
\end{document}